\documentclass[11pt]{article} 
\usepackage{silence} 
\usepackage[margin=1in]{geometry} 
\pdfoutput=1 

\usepackage[dvipsnames,svgnames,table]{xcolor} 
\usepackage{tikz}
\usetikzlibrary{arrows.meta}
\usetikzlibrary{decorations.pathreplacing}

\usepackage{amsmath,amsfonts,amssymb,amsthm}
\usepackage{mathtools}

\usepackage{mmap}
\usepackage[T1]{fontenc}

\usepackage[utf8]{inputenc}
\usepackage{csquotes}
\usepackage[english]{babel}

\usepackage[
    activate={true,nocompatibility},
    final, 
    tracking=true,
    kerning=true,
    spacing=true,
    factor=1100,
    stretch=10,
    shrink=10
    ]{microtype}
\microtypecontext{spacing=nonfrench}

\usepackage[backend=biber,
              style=alphabetic,
        maxbibnames=99,
      minalphanames=3,
      maxalphanames=4,
            backref=true]{biblatex}

\def\anon{0} 

\if\anon1
\usepackage{hyperref}
\hypersetup{
    pdftitle={Quantum algorithms for hypergraph simplex finding}, 
    pdfauthor={Anonymous submission}, 
    colorlinks=true, 
    linkcolor=blue, 
    citecolor=blue, 
    urlcolor=green 
}
\else
\usepackage{hyperref}
\hypersetup{
    pdftitle={Quantum algorithms for hypergraph simplex finding}, 
    pdfauthor={Shalev Ben{-}David and Zhiying Yu}, 
    colorlinks=true, 
    linkcolor=blue, 
    citecolor=blue, 
    urlcolor=green 
}
\fi

\AtEveryBibitem{
 \clearlist{address}
 \clearfield{date}
 \clearlist{location}
 \clearfield{month}
 \clearfield{series}
 \clearfield{pages}
 \clearlist{organization}
 \clearfield{number}
 \clearlist{language}

 \ifentrytype{book}{}{
  \clearlist{publisher}
  \clearname{editor}
  \clearfield{issn}
  \clearfield{isbn}
  \clearfield{volume}
 }
}

\DeclareFieldFormat{eprint:eccc}{
\mkbibacro{ECCC}\addcolon\space\ifhyperref
    {\href{http://eccc.hpi-web.de/report/#1/}{\nolinkurl{#1}}}
    {\nolinkurl{#1}}}
\DeclareFieldAlias{eprint:ECCC}{eprint:eccc}
\DeclareFieldFormat{eprint:iacr}{Cryptology ePrint\addcolon\space\ifhyperref
    {\href{https://eprint.iacr.org/#1}{\nolinkurl{#1}}}
    {\nolinkurl{#1}}}
\DeclareFieldAlias{eprint:IACR}{eprint:iacr}
\DeclareFieldFormat{eprint:arxiv}{arXiv\addcolon\space\ifhyperref
    {\href{https://arxiv.org/abs/#1}{\nolinkurl{#1}}}
    {\nolinkurl{#1}}}

\DeclareFieldFormat[article]{title}{#1}
\DeclareFieldFormat[inproceedings]{title}{#1}
\DeclareFieldFormat[incollection]{title}{#1}
\DeclareFieldFormat[online]{title}{#1}

\renewbibmacro{in:}{}

\newcommand{\lName}{1}

\newcommand{\donothing}[1]{#1}

\newcommand{\JACM}{\if\lName1\donothing{Journal of the {ACM}}\else{JACM}\fi}
\newcommand{\SICOMP}{\if\lName1\donothing{{SIAM} Journal on Computing}\else{SICOMP}\fi}
\newcommand{\ToC}{\if\lName1\donothing{Theory of Computing}\else{ToC}\fi}
\newcommand{\ToCGS}{\if\lName1\donothing{Theory of Computing Graduate Surveys}\else{ToC}\fi}
\newcommand{\TOCT}{\if\lName1\donothing{{ACM} Transactions on Computation Theory}\else{TOCT}\fi}
\newcommand{\ToIT}{\if\lName1\donothing{{IEEE} Transactions on Information Theory}\else{TOCT}\fi}
\newcommand{\CCjournal}{\if\lName1\donothing{Computational Complexity}\else{CC}\fi}
\newcommand{\CJTCS}{\if\lName1\donothing{Chicago Journal of Theoretical Computer Science}\else{CJTCS}\fi}

\newcommand{\TCS}{\if\lName1\donothing{Theoretical Computer Science}\else{TCS}\fi}
\newcommand{\IPL}{\if\lName1\donothing{Information Processing Letters}\else{IPL}\fi}
\newcommand{\JCSS}{\if\lName1\donothing{Journal of Computer and System Sciences}\else{JCSS}\fi}

\newcommand{\RSA}{\if\lName1\donothing{Random Structures and Algorithms}\else{RSA}\fi}
\newcommand{\JCTA}{\if\lName1\donothing{Journal of Combinatorial Theory, Series A}\else{JCTA}\fi}
\newcommand{\JCTB}{\if\lName1\donothing{Journal of Combinatorial Theory, Series B}\else{JCTB}\fi}
\newcommand{\PJM}{\if\lName1\donothing{Pacific Journal of Mathematics}\else{PJM}\fi}
\newcommand{\QICjournal}{\if\lName1\donothing{Quantum Information and Computation}\else{QIC}\fi}
\newcommand{\IJQI}{\if\lName1\donothing{International Journal of Quantum Information}\else{IJQI}\fi}
\newcommand{\PRA}{\if\lName1\donothing{Physical Review A}\else{PRA}\fi}
\newcommand{\PRL}{\if\lName1\donothing{Physical Review Letters}\else{PRL}\fi}
\newcommand{\VLDB}{\if\lName1\donothing{International Journal on Very Large Data Bases}\else{VLDB}\fi}

\DefineBibliographyStrings{english}{%
  backrefpage = {p{.}},
  backrefpages = {pp{.}},
}

\newtheorem{theorem}{Theorem}

\newtheorem{lemma}[theorem]{Lemma}

\newtheorem{corollary}[theorem]{Corollary}
\newtheorem{definition}[theorem]{Definition}

\newtheorem{open}{Open Problem}
\theoremstyle{definition}

\newcommand{\eq}[1]{\hyperref[eq:#1]{(\ref*{eq:#1})}}
\renewcommand{\sec}[1]{\hyperref[sec:#1]{Section~\ref*{sec:#1}}}
\newcommand{\thm}[1]{\hyperref[thm:#1]{Theorem~\ref*{thm:#1}}}
\newcommand{\lem}[1]{\hyperref[lem:#1]{Lemma~\ref*{lem:#1}}}
\newcommand{\defn}[1]{\hyperref[def:#1]{Definition~\ref*{def:#1}}}
\newcommand{\prop}[1]{\hyperref[prop:#1]{Proposition~\ref*{prop:#1}}}
\newcommand{\cor}[1]{\hyperref[cor:#1]{Corollary~\ref*{cor:#1}}}
\newcommand{\fig}[1]{\hyperref[fig:#1]{Figure~\ref*{fig:#1}}}
\newcommand{\tab}[1]{\hyperref[tab:#1]{Table~\ref*{tab:#1}}}
\newcommand{\alg}[1]{\hyperref[alg:#1]{Algorithm~\ref*{alg:#1}}}
\newcommand{\app}[1]{\hyperref[app:#1]{Appendix~\ref*{app:#1}}}
\newcommand{\conj}[1]{\hyperref[conj:#1]{Conjecture~\ref*{conj:#1}}}
\newcommand{\chap}[1]{\hyperref[chap:#1]{Chapter~\ref*{chap:#1}}}

\newcommand{\B}{\{0,1\}}

\DeclareMathAlphabet{\mathbbold}{U}{bbold}{m}{n}

\newcommand{\OR}{\mathtt{OR}}

\newcommand{\SF}{\mathtt{SF}}

\DeclareMathOperator{\D}{D}

\DeclareMathOperator{\cD}{\mathcal{D}}
\DeclareMathOperator{\R}{R}
\DeclareMathOperator{\Q}{Q}
\DeclareMathOperator{\C}{\mathcal{C}}

\newcommand{\ket}[1]{\left\vert #1\right\rangle}

\newcommand{\bigslash}[2]{\left. #1 \middle/ #2 \right.}

\renewcommand*{\S}{\mathcal{S}}

\renewcommand*{\O}{\mathcal{O}}
\renewcommand*{\P}{\mathcal{P}}
\newcommand*{\G}{\mathcal{G}}
\newcommand*{\V}{\mathcal{V}}
\newcommand*{\E}{\mathcal{E}}
\newcommand*{\F}{\mathcal{F}}
\newcommand*{\LG}{\mathcal{LG}}
\newcommand*{\bS}{\boldsymbol{S}}
\newcommand*{\bU}{\boldsymbol{U}}
\newcommand*{\bC}{\boldsymbol{C}}
\newcommand*{\bA}{\boldsymbol{A}}

\newcommand{\Exp}{\mathop{\mathbb{E}}}

\newcommand*{\permutate}[2]{P^{#1}_{#2}}
\newcommand*{\squareb}[1]{\left[#1\right]}
\newcommand*{\roundb}[1]{\left(#1\right)}
\newcommand*{\curlyb}[1]{\left\{#1\right\}}

\begin{document}

\title{Quantum algorithms for hypergraph simplex finding}

\if\anon1
\author{Anonymous submission}
\else
\author{
    Shalev Ben{-}David\\
    \small Institute for Quantum Computing\\
    \small University of Waterloo\\
    \small \texttt{shalev.b@uwaterloo.ca}
    \and
    Zhiying Yu\\
    \small Institute for Quantum Computing\\
    \small University of Waterloo\\
    \small \texttt{zy3yu@uwaterloo.ca}
}
\fi

\date{}
\maketitle

\begin{abstract}
We study the quantum query algorithms for simplex finding,
a generalization of triangle finding to hypergraphs.
We motivate this problem by showing it
satisfies a rank-reduction property: a quantum query algorithm
for finding simplices in rank-$r$ hypergraphs can be turned
into a faster algorithm for finding simplices in
rank-$(r-1)$ hypergraphs. In particular, we show that for any constant rank 
$r$, an $O(n^{r/2})$ quantum algorithm for finding a simplex
in rank-$r$ hypergraphs would imply an $O(n)$ quantum algorithm
for triangle finding. 


We then study two techniques used to design quantum query algorithms:
nested quantum walks on Johnson graphs, and adaptive learning graphs.
We show that every nested Johnson graph quantum walk (with any constant number of nested
levels) can be converted into an adaptive learning graph.
Along the way, we introduce the concept of $\alpha$-symmetric learning
graphs, which is a useful framework for designing and analyzing
complex quantum search algorithms.
Inspired by the work of Le Gall, Nishimura, and Tani (2016)
on $3$-simplex finding, we use our new technique to obtain an 
algorithm for $4$-simplex finding in rank-$4$ hypergraphs with $O(n^{2.46})$ 
quantum query cost, improving the trivial $O(n^{2.5})$ algorithm.
\end{abstract}

\if\anon{1}
{\footnotesize\tableofcontents}
\else
{\scriptsize\tableofcontents}
\fi
\clearpage


\section{Introduction}

A famous property of quantum algorithms is 
that they can be used to get polynomial speedups for unstructured
search problems, as shown by Grover \cite{Gro96}. 
Given query access to an array of size $n$
containing a marked item, Grover's algorithm finds the
marked item using only $O(\sqrt{n})$ quantum queries.
This algorithm can be used to find an easy-to-check certificate
using quadratically fewer queries than the number of possible
locations of the certificate.

However, when the search problem takes place in a richer combinatorial structure, 
    Grover's algorithm does not exploit the extra information.
For example, in the task of element distinctness,
we are given query access to an array $x$ of $n$ integers,
and we are asked to find a pair of positions $(i,j)$
such that $x_i=x_j$. 
Since the number of pairs is $\Theta(n^2)$,
Grover's algorithm uses the trivial $\Theta(n)$ queries,
which is not an improvement over querying all the input symbols.
However, an algorithm by Ambainis based on quantum walks
\cite{Amb07} achieves $O(n^{2/3})$ queries, which is known to be tight
\cite{AS04}. Other types of search problems have also
been studied, including $k$-sum \cite{BS13} and
$k$-distinctness \cite{BL11,Bel12a}; for the latter, the asymptotic
complexity of the best
possible quantum algorithm is not known for any constant
values $k > 2$.

We note that the decision version of the problem
(detect whether a marked item exists) is easily seen to be
equivalent, up to low-order terms,
to the search version of the problem (find a marked
item) when the certificate we are searching for is of constant
size. For this reason,
we will generally talk about the decision and search versions
interchangeably.

\subsection{Graph search problems}

Graph search problems define a particularly interesting class
of problems for the study of quantum algorithms.
In this setting, we are given query access to $\binom{n}{2}$ bits
representing the presence or absence of an edge in a graph
with $n$ vertices. The task is to detect the presence of
some substructure in the graph. The most famous example of a graph search 
problem is triangle finding: the task is to find three vertices 
$i,j,k$ such that the query value is $1$ on all pairs 
(that is, $x_{ij}=x_{jk}=x_{ki}=1$).

A large amount of work has been dedicated to determining the
quantum query complexity of triangle finding \cite{MSS07,San08,Bel12,LMS17,LGal14}.
These works culminated in a triangle finding algorithm
that uses $O(n^{5/4})$ quantum queries,
down from the trivial $O(n^2)$ cost of querying all the edges
(and down from the $O(n^{3/2})$ cost of applying Grover
search to the set of $\binom{n}{3}$ possible triangles).
In the lower bound direction, we only know the trivial 
$\Omega(n)$, which follows via a reduction from unordered search
combined with a lower bound such as \cite{BBBV97}
for the latter task.

In fact, no non-trivial lower bound is known for \emph{any}
graph search problem: given any constant-sized subgraph,
the best lower bound known for checking subgraph containment 
in an input graph on $n$ vertices is $\Omega(n)$, 
despite the fact that for larger subgraphs the best known 
upper bound approaches the trivial $O(n^2)$. 
The lack of good lower bounds is a consequence
of the \emph{certificate barrier} for the positive-weight
quantum adversary method \cite{SS06}.
This barrier says that the positive-weight quantum
adversary method (the main lower bound technique for
quantum query complexity) cannot be used to give
a lower bound better than $\Omega(\sqrt{\C_1(f)n})$
for any function $f$, where $\C_1(f)$ denotes the $1$-certificate
complexity of $f$ and $n$ denotes the input size.
In particular, when searching for certificates of constant
size, the best lower bound this technique can give
is the square root of the input size (i.e. $\Omega(n)$ for
graph problems, which have input size $\Theta(n^2)$).

Attempts to improve the best quantum algorithms for search
problems have led to new insights into the design of quantum
algorithms.
Two interesting and powerful techniques came up in previous work. 
One of them is quantum walks on a Johnson graph \cite{Amb07}, 
which was later generalized to nested quantum walks \cite{JKM13}.
The other technique is the framework of learning graphs \cite{Bel12},
which was later generalized to adaptive learning graphs \cite{CLM19}.
Most non-Grover quantum search algorithms (particularly for graph
search problems) use one of these frameworks in their design.
Details of these two computational frameworks are outlined in
\sec{preliminaries}.

\subsection{Our results}

With the aim to expand our understanding of quantum query complexity 
for search problems, we consider the generalization of graph search to
hypergraph search. 
The hyperedges of a rank-$r$ hypergraph
are elements of $\binom{[n]}{r}$ (that is, subsets of size $r$
of the set $[n]=\{1,2,\dots,n\}$).
We assume $r$ is constant, so the trivial algorithm
of querying everything uses $O(n^r)$ queries.
The generalization of triangle finding to hypergraphs is called 
    \emph{simplex finding}. In this task, we are searching for 
    an \emph{$r$-simplex}, which is a clique in an $r$-uniform 
    hypergraph consisting of $r+1$ vertices, each $r$ of which are 
    connected by a hyperedge. 
We denote the query task of simplex finding (with $n$ vertices
and rank $r$) by $\SF_{n,r}$.
Grover search can be used to show that $\Q(\SF_{n,r})=O(n^{(r+1)/2})$. 
For the lower bound, a simple reduction from unordered search can also 
    show $\Q(\SF_{n,r})=\Omega(n^{r/2})$ (see \lem{SFbounds}).
We call these the trivial bounds on $\Q(\SF_{n,r})$.

Intuitively, increasing the dimensionality of graphs should make it more difficult 
    to detect whether a specific substructure exists.
Despite this, the best-known lower bound for simplex finding remains trivial even 
    when the rank $r$ is large.
In the upper bound direction, we made the attempt to generalize existing 
    quantum algorithm design methodologies for graph search problems 
    to hypergraphs.
Using quantum walks that search for vertices of a certificate, one can derive
    nontrivial query upper bounds for $3$-simplex finding.
(Interesting quantum algorithms for $3$-simplex finding have also
been shown in previous work such as \cite{LGNT16}.)

However, as the rank $r$ continues to increase, simple algorithmic design no longer 
    provides advantages. 
Any query-efficient quantum algorithm needs to make full utilization of 
    the hypergraph structure and such an algorithm becomes messy and
    technically challenging to describe.
Yet there are still interesting 
things to be learned from the study of hypergraph search problems.
Our work starts with the following insight.

\begin{theorem}\label{thm:reduction}
An algorithm for simplex finding in rank $r+1$ hypergraphs
can be converted into a faster algorithm for simplex
finding in rank $r$ hypergraphs. That is,
\[\Q(\SF_{n,r})=O\left(\frac{\Q(\SF_{n,r+1})}{\sqrt{n}}\right).\]   
\end{theorem}

This theorem says, in particular, that an $O(n^{r/2})$ algorithm
for simplex finding in rank $r$ hypergraphs
(for any specific constant $r$)
implies an $O(n)$ algorithm for triangle finding. Conversely,
any non-trivial lower bound for triangle finding will give
non-trivial lower bounds for simplex finding in all higher rank
hypergraphs.

\thm{reduction} suggests that the study of simplex finding
in higher-rank hypergraphs is useful for 
graph search problems like triangle finding. An algorithm
for hypergraphs immediately implies an algorithm for graphs;
in the reverse direction, a non-trivial lower bound for
hypergraphs seems like a good first step towards a non-trivial
lower bound for graphs, because \thm{reduction} says that
the hypergraph lower bound is formally easier.

In an attempt to find good algorithms for higher rank simplex finding,
we investigate the framework of nested quantum walks and
adaptive learning graphs. We give a formal reduction between
them, showing that every nested quantum walk (on Johnson graphs)
can be converted into an adaptive learning graph.

\begin{theorem}[Informal; see \lem{learning_graph_nested_johnson}]
\label{thm:conversion}
Every nested quantum walk on Johnson graphs can be converted
into an adaptive learning graph for the same task and with the
same query cost, so long as the ``checking'' step of the
walk can itself be implemented with a learning graph.
\end{theorem}

Finally, for symmetric yet complex problems such as simplex finding
in hypergraphs, we introduce the notion of ``$\alpha$-symmetric
learning graphs,'' an adaptive learning graph built iteratively
from ``stages'' of a special type.
Using this framework, we
show the following theorem, which provides a nontrivial algorithm 
for $4$-simplex finding.

\begin{theorem} \label{thm:4_simplex_algorithm}
There is an adaptive learning graph algorithm that computes $4$-simplex
    finding with $O(n^{2.455})$ quantum queries.
\end{theorem}

Our algorithm is somewhat complex, so describing and analyzing
it in terms of nested quantum walks would be prohibitively difficult.
Instead, or framework of ``$\alpha$-symmetric learning graphs'' (which
are still closely motivated by nested quantum walks) lets us abstract
away some of the details and makes the analysis more tractable.

Furthermore, this framework has the potential to generalize, giving a possible
    direction for nontrivially solving large-rank hypergraph search 
    problems in more generality.
Our work has the potential to be useful in the design and analysis of 
    complicated quantum search algorithms.

\subsection{Our techniques}

\subsubsection*{Rank reduction for simplex finding}
We give a randomized reduction which reduces the task
$\OR_n\circ \SF_{n,r}$ to $\SF_{2n,r+1}$.
The former is the task of determining whether there is
an $r$-simplex in any of $n$ given $r$-hypergraphs of $n$
vertices each; even though this search involves
$n^2$ vertices, we show that we can complete this search
using a hypergraph search with one additional rank
(i.e. rank $r+1$) on only $2n$ vertices.

Given $n$ $r$-hypergraphs of $n$ vertices, we identify
all their vertex set with the same set of vertices $B$;
Having another set of $n$ vertices, $A$, each of which
is used to label one of the hypergraphs.
We construct an $r+1$-hypergraph $G$ consists of $2n$ vertices
$A\cup B$. Its edges will be as follows: for every
input hypergraph $G_v$ labeled by $v\in A$, and for each
hyperedge $e$ of $G_v$, we add the hyperedge $\{v\}\cup e$
to $G$. In other words, $G$ will have $2n$ vertices
and the same number of hyperedges as the total of all $n$
input hypergraphs; it will have rank $r+1$ if the input
hypergraphs have rank $r$.

We want to run a simplex-finding algorithm on $G$
to find a simplex of one of the $G_v$ graphs, but this does not
yet work. That's because
an $(r+1)$-simplex in $G$ does not quite correspond to an $r$-simplex
in one of the input hypergraphs $G_v$. Indeed, recall that
an $(r+1)$ simplex in $G$ is a set of $r+2$ vertices and
$r+2$ hyperedges; of the hyperedges, $r+1$ of them contain
a vertex $v\in A$, and deleting $v$ from these hyperedges
gives an $r$-simplex in $G_v$, but the last hyperedge
corresponds to $r+1$ vertices in $B$ and no vertices in $A$,
which cannot occur at all in our graph $G$.

To solve this issue, we need to add hyperedges within the
vertex set $B$ of $G$. However, we wish to do so without
forming a simplex in $G$ that uses only vertices from $B$.
To this end, we randomly partition $B$ into $r+1$ parts,
and add all the hyperedges in the ``complete $(r+1)$-partite
hypergraph'' (i.e. all sets of $r+1$ vertices that use
exactly one vertex from each part of $B$). This ensures
we did not introduce a simplex in $G$ that uses only
vertices of $B$ (such a simplex would need to have $r+2$
vertices, and hence two vertices would lie in the same
part of the partition, which is impossible).

We then run simplex-finding on the modified hypergraph $G'$
with these extra hyperedges. For any simplex of some
input graph $G_v$, the simplex will give rise to a higher-rank
simplex in $G'$ if and only if each of the $r+1$ vertices
of the simplex is in a different part of the partition of $B$;
this happens with constant probability, as the number of partitions is constant.
Repeating this search constantly many times with different
partitions of $B$ will result in finding a simplex in one of the
input graphs $G_v$ with high probability.

This reduction used a single copy of $\SF_{2n,r+1}$ to solve
the $\OR$ of $n$ copies of $\SF_{n,r}$. The latter task
requires $\Theta(\sqrt{n})$ times as many quantum queries
as $\SF_{n,r}$, since bounded-error quantum query complexity composes
multiplicatively \cite{HLS07,Rei11,LMR+11,Kim13}. We note that the proof
of the latter fact uses the negative-weight-adversary characterization
of quantum query complexity, so our reduction
technically uses the negative-weight adversary method.

\subsubsection*{Converting nested quantum walks into learning graphs}

It was shown by \cite{CLM19} that
a simple quantum walk on the Johnson graph can be converted to an
equivalent algorithm formulated in the adaptive learning graph framework
with equivalent quantum query cost (see \lem{learning_graph_johnson_walk}).
However, many query algorithms use a ``nested'' quantum walk,
in which one quantum walk occurs as a subroutine of another;
this case is not handled by the construction in \cite{CLM19}.
Our objective is to show that a nested quantum walk on Johnson graph can also be converted to an  equivalent learning graph algorithm.

We will focus on the variant of an $r$-level nested quantum walk presented by 
Jeffery, Kothari, and Magniez \cite{JKM13}.
In this version, we only keep one data structure in quantum registers 
$D(A_1, \dots, A_r) = \ket{A_1, \dots, A_r, D(A_1, \dots, A_r)}$,
which keeps track of the state of all quantum walk levels and initialized at the computation's beginning.
This allows setup costs to appear only at the beginning of the computation.
The updates of each quantum walk level proceed to act on the state
$D(A_1, \dots, A_r)$ instead of their individual classical data structure.

Recall that the Johnson graph $J(n, k)$ has vertices in $\binom{[n]}{k}$
and two vertices $A, B$ are connected by an edge if they differ by
exchanging exactly one element.
Usually, the quantum walk on the Johnson graph is symmetric, meaning that we
    designate $\ell$ elements in $[n]$ as certificates and define the marked vertices of $J(n, k)$
    as all $A \in \binom{[n]}{k}$ where $A$ contains all the certificates.
This allows us to build a corresponding learning graph with special
symmetric stages.

Let's assume the $r$-layers of Johnson walk are given by 
$\{J(n_i, k_i)\}_{i\in [r]}$, where the walk on $J(n_{i+1}, k_{i+1})$ 
    appears as the checking procedure of the walk on $J(n_{i}, k_{i})$.    
To formulate an equivalent adaptive learning graph, we mimic the
setup-update-checking procedures 
in the original algorithm and build their respective stages.
The learning graph begins with $r$ levels of setup stages, loading the states $A_1, \dots, A_r$ respectively.
It's followed by $\sum_{i \in [r]}\ell_i$ stages, loading the certificates of $A_1, \dots, A_r$  in the given order.
The final stage defines the checking procedure of the innermost quantum walk.

Some important modifications must be made to
the proof of \lem{learning_graph_johnson_walk} when extending
it to nested Johnson walks.
In the standard learning graph definition,
loaded elements are kept in an unordered set.
The first issue comes up when the state space of an inner quantum walk needs
to rely on the state of an outer walk.
Stacking stages naively doesn't provide such a dependency.
Instead, we label the vertices of the learning graph by ordered partial
subsets $\P(X, k)$ instead of unordered sets $\binom{X}{k}$.

A related modification concerns the certificates in the learning graph.
Let $y$ be a $1$-input to the learning graph.
The certificate of the nested quantum walk is given by a sequence of certificates 
at each level, $I_y = (I_{y, 1}, \dots, I_{y, r})$.
Even if the certificate for $y$ is unique,
for each $i\in [r]$, elements of the certificate 
can appear in different positions of the ordered tuple in $\P([n_i], k_i)$.
We therefore set each $I_{y,i}$ to refer to the indices of certificates
in the outer levels; this contains information regarding both what
the certificate is and where it is found within each ordered tuple.
Fortunately, this information is available during the setup stages,
so this modification does not pose problems for an adaptive learning graph.

\subsubsection*{Learning graphs with $\alpha$-symmetric stages}

We introduce the concept of an $\alpha$-symmetric learning graph,
which is a special type of a learning graph which is easier to design and
can capture most of the known algorithms for graph search problems.

The motivation for $\alpha$-symmetric learning graphs is an issue
that came up when designing a 4-simplex-finding algorithm.
For the intermediate stages to be well-defined, we require the marked states 
    to not only contain the certificates, but to satisfy certain degree requirements.
In rare cases, the learning graph may load vertices whose degree becomes
too large; we wish to remove such vertices.
In this case, the stages we design are no longer fully symmetric, but they are not too far
    from being symmetric.
We define $\alpha$-symmetric stages of an adaptive learning graph to 
    capture this scenario.

Assume that $\alpha$ is an exponentially small (with respect to $n$) fraction,
    and let $s$ be a constant.
An $\alpha$-symmetric stage in a learning graph is a stage that can be obtained
from a fully-symmetric stage $\F$ by slightly altering its flows.

More concretely, let $V_{s, y}, V_{s+1, y}$ be the beginning and ending vertex sets 
    of $\F$, respectively, which receive positive flow from the flow $p_y$ of a $1$-input
    $y$.
We identify a $(1 - (1 - \alpha)^s)$ fraction of $V_{s, y}$
    and a $(1 - (1 - \alpha)^{s+1})$ fraction 
    of $V_{s+1, y}$ as ``bad'' or ``unavailable''.
If any bad vertices receive or emit positive flows from $p_y$ in $\F$, 
    we delete the flows on this vertex (by removing the edges incident to the 
    bad vertex or by setting the flow on those edges to $0$) and redistribute the 
    flows evenly to the remaining vertices in $V_{i, y}$ or $V_{j, y}$.
In symmetric stages, vertices in $V_i$ receive uniform flow, meaning 
    $p_y(v)$ are equal for all $v \in V_{i, y}$.
In an $\alpha$-symmetric stage, flow values are allowed to differ but 
    are close to each other.
In particular, we get $\displaystyle p_y(v) \le \frac{1}{(1-\alpha)^s}p_y(w)$ 
    for any $v, w \in V_{i, y}$.

We can design learning graphs by stacking $\alpha$-symmetric stages just like 
    stacking symmetric stages because the given construction allows the 
    ending vertices of an $\alpha$-symmetric stage with constant $s$ to act 
    as the beginning vertices of an $\alpha$-symmetric stage with constant $s+1$.
If the learning graph is designed with a fixed number of levels, $s$ refers to the 
    stage number and it is upper bounded by a constant. 
Therefore, when $\alpha$ is small, this redistribution doesn't alter the 
    asymptotic bound of the algorithm's query complexity.

\subsubsection*{Simplex finding in rank \texorpdfstring{$4$}{4}}

In rank $4$ hypergraphs, the best-known quantum algorithm for
simplex finding algorithm has the trivial $O(n^{2.5})$ query upper bound.
Inspired by the approach used by Le Gall, Nishimura, and Tani \cite{LGNT16} when building the 
    (current) optimal $3$-simplex finding algorithm, we will show in 
    \sec{4_simplex_finding} that a nontrivial algorithm 
    can be constructed by a nested quantum walk on $30$ levels of nested Johnson graphs,
    searching for the ``hyperedges'' of rank $1$, $2$, $3$, $4$, in order.
The algorithm uses $30$ parameters $a_i, b_{ij}, c_{ijk}, d_{ijkl}$ for
    $ijkl \in \binom{[5]}{4}$, used to set up the size of the state of the Johnson walks.
    
If the input $4$-uniform hypergraph $G$ has a $4$-simplex with vertices $u_1, \dots, u_5$,
    the first $5$ levels will each have one of these vertices as the certificate.
The state of these quantum walks is labeled by $A_i$, with size $n^{a_i}$.
The next $10$ levels (levels $6$ through $15$) will search for pairs of vertices $u_{ij}$, 
    via a quantum walk with state labeled by $B_i$ in the smaller state space
    $\Gamma_{ij} = A_i \times A_j$.

Similarly, in levels $16$ to $25$, we search for the triples of vertices
    $u_{ijk}$ by a quantum walk over the state $C_{ijk}$.
However, for these $10$ levels, the state space $\Gamma_{ijk}$
    we are walking on is the set of 
    triples of vertices $v_iv_jv_k$ where $v_iv_j \in B_{ij}$, $v_iv_k \in B_{ik}$, and $v_jv_k \in B_{jk}$.
In other words, we only consider walking on the $2$-dimensional face,
    finding a triangle for which all of the lower-rank hyperedges
    were already found at the earlier levels of the search.
The expected size of this state space is smaller than the trivial state space 
    $A_i \times A_j \times A_k$, which is critical for making the resulting algorithm 
    nontrivial.

Note that given arbitrary states $B_{ij}, B_{ik}, B_{jk}$, the size of $\Gamma_{ijk}$ may vary.
We want to avoid the case that $\Gamma_{ijk}$ has size larger than some constant multiple of 
    its expected size $O(n^{m_{ijk}}), m_{ijk} = b_{ij} + b_{ik} + b_{jk} - a_i - a_j - a_k$.
We can achieve this by controlling the degrees of the $2$-edges found in levels $6$ to $15$.
Thus, by adding appropriate degree constraints to marked elements in levels
    $6$ through $15$ of the nested quantum walk,
    we ensure a smaller state space $\Gamma_{ijk}$. 
These extra degree constraints fails with exponentially small probability,
    but we can handle this in the framework of $\alpha$-symmetric learning graphs.

Finally, in the last $5$ stages, we search for the five hyperedges of $4$-simplex
    by quantum walking on 
    the state space $\Gamma_{ijkl}$ consisting of the $3$-dimensional polytope
    (i.e. $4$-hyperedges)  whose geometric faces are already found at the earlier levels.
Adding degree constraints to marked elements in levels $16$ to $25$ ensures these 
    quantum walks have good complexities.
Analyzing the query complexity of this learning graph and linearly optimizing the $30$ parameters
    provide a nontrivial $O(n^{2.455})$-query quantum algorithm for $4$-simplex finding.

It is important to note that although we described the algorithm as a nested quantum walk,
we formally present it as an adaptive learning graph using our
$\alpha$-symmetric framework; this presentation
makes the analysis of the algorithm more tractable, demonstrating the utility
of the framework.

\subsection{Open problems}

One of the main open problems for graph search problems
    is the long-standing task of finding a 
    non-trivial lower bound for triangle finding.
As \thm{reduction} shows, a formally easier version of this
problem is to find a non-trivial lower bound for simplex finding.

\begin{open}
Is there an $\Omega(n^{r/2+0.01})$ lower bound for
simplex finding in any rank $r$?
\end{open}

We are also interested in understanding how the complexity
of simplex finding increases with $r$.

\begin{open}
Let $a_r=\inf\{a:\Q(\SF_{n,r})=O(n^{r/2+a})\}$.
We know by \thm{reduction} that
\[0\le a_2\le a_3\le\dots \le 1/2.\]
Is this sequence strictly increasing?
What is $\lim_{r\to\infty} a_r$?
\end{open}

More specifically, an interesting problem is whether \thm{4_simplex_algorithm}
generalizes to higher-rank hypergraphs; if it can be made to give nontrivial
for all $r$, this would at least imply that $a_r<1/2$ for every $r$.

\begin{open}
Can \thm{4_simplex_algorithm} be generalized to a nontrivial quantum algorithm
for $r$-simplex finding, for all $r\ge 4$?
\end{open}

Finally, one can ask similar questions for other families of subgraph finding
problems.

\begin{open}
Can our techniques be used to find new algorithms for other (hyper)graph
search problems? Are there other reductions between natural
families of (hyper)graph search problems, similar to \thm{reduction}?
\end{open}


\section{Preliminaries} \label{sec:preliminaries}

\subsection{Hypergraph Notations}

We start by introducing some notations for hypergraphs.
A hypergraph $G$ consists of a set of vertices $V$ and a set of
hyperedges $E$, where every hyperedge $e \in E$ is a subset of $V$.
We call a hyperedge $e$ with $k$ elements a $k$-edge, where $k$ is the \textit{size} of $e$.
For the problems presented in this paper, we assume that the
hypergraphs have no parallel hyperedges and no hyperedges of size $0$ or $1$.

We use $n$ to denote the size of $V$.
The \textit{rank} $r$ of a hypergraph $G$ is the size of the 
    largest hyperedge in $E$.
We only consider the rank $r$ as a constant relative to $n$.
Furthermore, if every hyperedge of $G$ has size $r$, we call $G$ an 
$r$-uniform hypergraph or an $r$-hypergraph in short.

Let $[n]$ denote the set $\{1, 2, \dots, n\}$, and
let $\permutate{n}{r} = \bigslash{n!}{(n-r)!}$.
To denote an element in $\binom{V}{r}$ conveniently, we often omit 
    the curly bracket of a set.
For example, we write a potential hyperedge $\{u, v, w\} \in \binom{V}{3}$
    simply as $uvw$.

Given a hypergraph $G$ and subsets $A, B\subseteq V$, we use $G_A$ 
to denote the restriction of $G$ to $A$ (i.e. the subgraph of $G$ 
induced by $A$) and we use $G_{A_1, \dots, A_r}$ to denote the 
$r$-partite hypergraph obtained 
from taking the restriction of $G$ to the $r$-partition 
$A_1, \dots, A_r$.

Observe that a graph is a 2-uniform hypergraph, so some graph terminology generalizes to hypergraphs.
Given a hypergraph $G = (V, E)$,
we say $v, w \in V$ are \emph{adjacent} if $v \neq w$ and there is a hyperedge $e \in E$ such that 
    $\{v, w\} \subseteq e$, two hyperedges $e_1, e_2 \in E$ are \emph{adjacent} if $e_1 \cap e_2 \neq \emptyset$.
We say that $v \in V$ is \emph{incident} to $e\in E$ if $v \in e$.
The \emph{degree} of a vertex $v\in V$ is the number of hyperedges incident to it.
With the above definitions, the concept of isomorphism and
of an incidence matrix naturally extend to hypergraphs.

Let $G = (V, E)$ be an $r$-uniform hypergraph.
The ($r$-dimensional) \emph{adjacency tensor} is a function 
    $f_G : \binom{V}{r} \to \{0, 1\}$ where 
    $f_G(v_1v_2\dots v_r) = 1$ if and only if $v_1v_2\dots v_r \in E$.
If $f_G$ is the constant $0$ function,
we say $G$ is an \emph{empty} hypergraph.
If $f_G$ is the constant $1$ function,
we say $G$ is a \emph{complete} hypergraph.

For this paper, we focus on finding query algorithms for $r$-uniform 
    hypergraph problems.
This means we fix the set of vertices $V$ and treat $f_G$ as a 
    black box oracle input.
We usually set $V = [n]$ for convenience.
In a query algorithm, we rely on the ability to ask the hyperedge 
    oracle $O_G = f_G$ whether an element in $\binom{V}{r}$ 
    is a hyperedge of $G$ to determine whether $G$ has a certain property.
The query model is formalized in the next subsection.

\subsection{Query complexity}

In query complexity, we are interested in the task
of computing a Boolean function $f\colon[q]^N\to[M]$.
Here $[q]$ is an input alphabet, usually $\B$, and $[M]$
is an output alphabet, also usually $\B$.
We may allow $f$ to be \emph{partial} in the sense that $f$ can be 
    only defined on a subset $\cD \subseteq [q]^N$.
We will use $\cD$ to denote the domain of $f$ (also called
a promise, since the input is promised to be in $\cD$).
If we restrict $f$ to a promise,
computing $f$ can only become easier because there are fewer
inputs to handle.
If $f$ is defined for all $x \in \{0, 1\}^N$, we say that $f$ 
is \emph{total}.

In the classical query model of computation, the input
$x\in\B^N$ (or $x\in[q]^N$)
is given as a black box oracle $\O_x$, which returns the
bit $x_i\in\B$ (or $x_i\in[q]$) given a query $i\in[N]$.
The goal is to find an algorithm which computes the value of $f(x)$
correctly with as few oracle calls to $\O_x$ as possible,
and succeeds on all inputs $x$ in the domain of $f$.

We make the following definitions.
\begin{itemize}
\item The \emph{deterministic query complexity} $\D(f)$
    of a (possibly partial) Boolean function $f$ is the minimum
    number of deterministic queries to an input $x$ that are required
    to compute $f(x)$ in the worst case over choice of $x$.
\item The \emph{randomized query complexity} $\R(f)$
is the minimum number $T$ such that there is a randomized
algorithm which makes $T$ queries in the worst case
and computes $f(x)$ to bounded error for all inputs $x$.
\item The \emph{quantum query complexity} $\Q(f)$ is the
minimum number $T$ such that there is a quantum algorithm
which makes at most $T$ queries (in superposition) and computes
$f(x)$ to bounded error for all inputs $x$.
\end{itemize}

For more detailed versions of these definitions, see \cite{BdW02}.
We note that randomized and quantum query complexities can
be amplified, so the probability of error achieved when
computing $f(x)$ does not matter so long as it is at most
a fixed constant in $(0,1/2)$.

Quantum query algorithms may take exponentially fewer queries to compute some partial functions than classical algorithms.
However, the hypergraph search problems 
we consider in this work are mostly total functions,
and the best separation between classical and quantum query
    complexity for total functions is at most polynomial: 
\begin{theorem}[\cite{BBC+01,ABK+21}]
For all total Boolean functions,
$\D(f) = O(\Q(f)^4)$. 
\end{theorem}

The following are important notions in query complexity.
\begin{itemize}
    \item A \emph{partial assignment} is a string
    $p\in\{0,1,*\}^N$ representing partial knowledge
    of a string in $\B^N$. We say two partial assignments $p$ and $q$
    are \emph{consistent} if for all $i\in[N]$ such that
    $p_i\ne *$ and $q_i\ne *$, we have $p_i=q_i$.
    We conflate a partial assignment $p$ with
    the set $\{(i,p_i):i\in[N], p_i\ne *\}$, which is a partial
    function from $[N]$ to $\B$. This lets us use notation
    such as $|p|$ for the number of non-$*$ bits of $p$.
    \item A \emph{certificate} for a (possibly partial)
    Boolean function $f$ is a partial assignment $c$ such that
    all inputs in the domain of $f$ which are consistent with $c$
    have the same $f$-value. In particular, a $1$-certificate
    has the property that $f(x)=1$ for all $x$ consistent with $c$,
    while a $0$-certificate has $f(x)=0$ for all $x$ consistent
    with $c$.
\end{itemize}

We also note a result on the quantum complexity of the composition of Boolean functions.
Let $f\colon \B^N \to \B$ and $g\colon\B^M \to \B$ be Boolean functions.
We define the composition
$f \circ g = f\circ (g, g, \dots, g) \colon \B^{NM} \to \B$ as the
function 
\[f \circ g(x^1x^2\dots x^N) := f(g(x^1),g(x^2),\dots ,g(x^N))\]
    for $x^1, x^2, \dots, x^N \in \{0, 1\}^M$.
A seminal result is that the quantum query complexity of the composed
function $f\circ g$
is equivalent to the product of quantum query complexities of $f$ and $g$.
\begin{theorem}[\cite{HLS07,Rei11,LMR+11,Kim13}] 
    \label{thm:composition_bound}
For any (possibly partial) Boolean functions $f$ and $g$, we have
\[ Q(f\circ g) = \Theta(Q(f)\cdot Q(g)).\]
\end{theorem}

\subsection{Quantum walks}

Quantum walks are a powerful tool in the design of quantum algorithms.
For our purposes, their main utility comes from their ability
to find marked vertices in a graph. See \cite{San08} for a survey.
Briefly, they are defined as follows. Let $P$ be an $n\times n$
stochastic matrix representing an ergodic, reversible Markov chain.
Let $\delta>0$ be the spectral gap of $P$, and let $\pi$ be its
unique stationary distribution.
We associate with every vertex $x\in[n]$ a data structure $D(x)$.
We assume we have access to
three quantum subroutines called setup, update, and checking;
the cost of the quantum walk
(i.e. the number of queries before a marked vertex
is found) will depend on their costs, which are defined as follows.
\begin{enumerate}
    \item Setup Cost $\bS$: The cost of setting up the initial state of the walk:
        \[\ket{S} = \sum_{x\in V}\sqrt{\pi_x}\ket{x, D(x)}\ket{0}.\]
    \item Update Cost $\bU$: The cost of making one step of transition:
        \[\ket{x, D(x)}\ket{0} \mapsto \ket{x, D(x)}\sum_{y\in V}\sqrt{P_{xy}}\ket{y, D(y)}.\]
    \item Checking Cost $\bC$: The cost of a quantum procedure checking if $x \in M$ using the data structure $D(x)$:
        if $x$ is marked, apply a $-1$ phase to the state $\ket{x, D(x)}$. 
\end{enumerate}
Then we have the following result.

\begin{theorem}[\cite{MNRS11}] \label{thm:mnrs_walk}
Let $P$ be an ergodic, reversible Markov Chain.
Let $\epsilon > 0$ be a lower bound on the probability that an element 
chosen from the stationary distribution $\pi$ of $P$ is marked.
Let $\delta > 0$ be the spectral gap of $P$.
Then there is a quantum algorithm that finds a marked vertex
with constant probability and 
\[O\left(\bS + \frac{1}{\sqrt{\epsilon}}\left(\frac{1}{\sqrt{\delta}}\bU
+ \bC\right)\right)\]
queries. In other words, we need to search for $O(1/\sqrt{\epsilon})$
steps, and each step costs $\bC$ for checking and $\bU/\sqrt{\delta}$
for walking.
\end{theorem}

For the design of quantum query algorithms for search problems,
such as Ambainis's algorithm for element distinctness \cite{Amb07},
we generally just need to walk on the Johnson graph.

\begin{definition}
    For $1\le k\le n/2$,
    The \emph{Johnson graph} $J(n, k)$ is the graph with vertex set
    $V=\binom{[n]}{k}$.
    Two vertices $A, B \in V$ are joined by an edge if and only if
    $|A\cap B|=k-1$,
     i.e.\ we can obtain $B$ from $A$ by removing an element of $A$ and adding a new element in $[n]$.
\end{definition}

The symmetric walk on $J(n, k)$ is given by a chain $P$ where
$P_{A, B} = k^{-1}(n-k)^{-1}$ for all $A, B$ adjacent in $J(n, k)$.
We note that $P$ is ergodic, reversible with stationary distribution $\pi$ equal to a vector of all $1/n$. The spectral gap of $P$ is
$1/k+1/(n-k)=\Theta(1/k)$.
Suppose that for some $\ell < k$, $A \in \binom{[n]}{k}$ is marked
if and only if $A$ contains a fixed subset of vertices 
    $v_1, \dots, v_\ell \in [n]$.
Then the fraction of marked states is
\[
\bigslash{\binom{n-\ell}{k-\ell}}{\binom{n}{k}} =
\Omega\left(\left(\frac{k-\ell}{n}\right)^{\ell}\right).\]
It is not hard to see that this is lower bounded by $\Omega((k/n)^\ell)$
when $\ell=O(\sqrt{k})$.
\begin{corollary} \label{cor:johnson_walk}
Let $k\le n/2$ and let $\ell=O(\sqrt{k})$.
Let $P$ be the symmetric Markov chain on $J(n,k)$,
and assume a vertex of $J(n,k)$ is marked if it contains
all of $\ell$ special elements in $[n]$.
Then the quantum walk algorithm finds a marked vertex of the Johnson
graph 
with constant success probability using
$O(\bS + (n/k)^{\ell/2}(\sqrt{k}\cdot \bU + \bC))$ queries.
\end{corollary}

Quantum walks on Johnson graphs are a key technique used to
construct nontrivial algorithms for graph search problems
such as triangle finding.

\subsection{Learning graphs}

In this subsection, we define the learning graph computational 
framework.
A feasible learning graph for Boolean function $f$ provides an upper 
bound to the quantum query complexity of $f$.

\subsubsection*{Basic learning graphs}

\begin{definition}[\cite{Bel12}] \label{def:learning_graph}
    Let $f$ be a Boolean function with domain $\cD \subseteq \{0, 1\}^N$.
    A \emph{(reduced) non-adaptive learning graph for $f$} is a directed acyclic graph $\G = (\V, \E)$ such that
    \begin{enumerate}
        \item every vertex $v \in \V$ is labeled by a subset $s(v) \subseteq [N]$ of indices of inputs to $f$,
        \item $\G$ has a root vertex labeled by the empty set $\emptyset$,
        \item every directed edge $e = \overrightarrow{uv} \in \E$ satisfies $s(u) \subseteq s(v)$,
        \item every directed edge $e = \overrightarrow{uv} \in \E$ has a length given by $l(e) = |s(v) - s(u)|$,
        \item every directed edge $e = \overrightarrow{uv} \in \E$ has a positive weight $w(e) \in \mathbb{R}^+$,
        \item every 1-input $y$ of $f$ (that is, $y \in f^{-1}(1)$) has a flow $p_y$ of value 1 on the 
            learning graph $\G$ where the root vertex of $\G$ is the source and every vertex $v \in \V$ such that 
            $s(v)$ contains a 1-certificate of $y$ in $f$ is a sink.
    \end{enumerate} 
\end{definition}

In order to distinguish the vertices and edges of a learning graph from the vertices and edges of a graph in the 
    question, we call the vertices in the learning graphs \textit{L-vertices} and call the directed edges in the 
    learning graphs \textit{L-edges} (or transitions).

In a learning graph, the label $s(v)$ of an L-vertex $v$ can 
be thought of as the set of oracle entries $\{(i,x_i) : i\in s(v)\}$
which are known to the algorithm if the algorithm is in the state
$v$; the graph itself gives a diagram of how the algorithm learns
the oracle entries.
We call $s(v)$ the set of \textit{loaded elements} of the L-vertex $v$
and we say an L-edge $e = \overrightarrow{uv}$ \textit{loads} elements 
$u_1, \dots, u_k$ if $s(v) - s(u) = \{u_1, \dots, u_k\}$.
Note that the graph does not depend on the input $x$,
but there is a flow for each $1$-input which does depend on the input;
such a flow specifies the (fractional) path taken by the algorithm
from the root (where it knows none of the oracle) to the sinks
(where it knows a $1$-certificate for the input).
The learning graph $\G$ is called ``non-adaptive'' because the L-edges
and their weights are independent of the input to the function.

\begin{definition} \label{def:lg_complexity}
    Let $\G = (\V, \E)$ be a non-adaptive learning graph for $f$.
    For $\F \subseteq \E$, the \textit{negative complexity} and \textit{positive complexity} of $\F$ is given by 
    \begin{equation} \label{eqn:lg_complexity}
        C_0(\F) := \sum_{e \in \F}l(e)w(e), \quad C_1(\F, y) := \sum_{e\in \F}l(e)\frac{p_y(e)^2}{w(e)}, \quad
        C_{1}(\F) := \max_{y \in f^{-1}(1)}C_1(\F, y).
    \end{equation}
    The \textit{learning graph complexity} of $\G$ is $\LG(\G) = \sqrt{C_0(\E)C_1(\E)}$.
    The \textit{learning graph complexity} $\LG(f)$ of the function $f$ is the minimum complexity of a 
        learning graph for $f$.
\end{definition}

A learning graph $\G$ can be turned into a feasible solution of
    the generalized adversary bound with objective value $\LG(\G)$ \cite{BL11}.
Therefore, every learning graph $\G$ for $f$ corresponds to a quantum query algorithm for $f$.
\begin{theorem} \label{thm:learning_graph_upper}
    For any (possibly partial) Boolean function $f$, $\Q(f) = O(\LG(f))$.
\end{theorem}

\subsubsection*{Conventions for designing learning graphs}

Here are some conventions for designing a learning graph $\G$ for
a function $f$.
Define the $i^{\text{th}}$ level of $\G$ by the set of L-vertices at depth $i$ from the root vertex of $\G$.
A stage of $\G$ will be the set of L-edges between level $i, j$ for some $i < j$.
Usually, the stages we are going to define only has depth 1, that is, $j = i + 1$.
We design learning graph by giving L-edges in stages.
Following the convention of \cite{CLM19}, we assume the 1-complexity of a stage $\F \subseteq \E$ is always 
upper bounded by 1; this can be achieved by multiplying the weights of every $e\in \F$ by $C_1(\F)$.
\begin{definition} \label{def:symmetric_stage}
    Suppose $\F$ is a stage with starting L-vertices $V_i$ and ending L-vertices $V_j$.
    Let $c := |V_i|, e := |V_j|$.
    We say $\F$ is \emph{symmetric} if
    \begin{itemize}
        \item every $v \in V_i$ has outdegree $d$ in $\F$,
        
        \item the number $c'$ of $v\in V_i$ that receives positive flow from $p_y$ is independent of $y\in f^{-1}(1)$,
            and the value of these positive flows all equal to $1/c'$,
        
        \item for every $v\in V_i$ that receives positive flow from $p_y$, $d'$ of the $d$ out-edges of $v$ have
            positive flow of equal values, the value $d'$ is independent of $y\in f^{-1}(1)$,

        \item the number $e'$ of $w\in V_j$ that receives positive flow from $p_y$ is independent of $y\in f^{-1}(1)$,
        and the value of these positive flows all equal to $1/e'$.
    \end{itemize}
\end{definition}

\noindent Let $\displaystyle T = \frac{cd}{c'd'}$ be the \emph{speciality} of $\F$. we get the following complexity 
    for $\F$.
\begin{lemma}[\cite{LMS17,CLM19}] \label{lem:symmetric_stage}
    Let $\F$ be a symmetric stage of $\G$ with speciality $T$.
    For every $y \in f^{-1}(1)$, if $L$ is the average length of the L-edges receiving positive flow then
        the L-edges in $\F$ can be weighted so that 
        $$C_0(\F) \le T\cdot L^2 \quad \text{ and } \quad C_1(\F, y) \le 1.$$
\end{lemma}


\subsubsection*{Adaptive learning graphs}

In an \emph{adaptive} learning graph, the weight of an L-edge may depend on queried entries 
    of the input $z$ to $f$.
\begin{definition}[\cite{CLM19}] \label{def:adaptive_learning_graph}
    Let $f$ be a (possibly partial) Boolean function with domain $\cD \subseteq \{0, 1\}^N$. 
    A directed acyclic graph $\G = (\V, \E)$ is an \emph{adaptive learning graph} for $f$ if it satisfies 
        all properties (1) to (6) in \defn{learning_graph}, except we replace property (5) with
    \begin{itemize}
        \item[5'.] For every $z \in \cD$ and directed edge $e = \overrightarrow{uv} \in \E$, there is a 
            positive weight value $w_{z_{s(v)}}(e) \in \mathbb{R}^+$, whose value depends only on $e$ and
            the loaded $s(v)$-entries of the input $z$.
    \end{itemize}
\end{definition}

Since $v$ is clear given the directed edge $e$, we abbreviate $w_{z_{s(v)}}(e)$ by $w_z(e)$.
The corresponding complexity of an adaptive learning graph is given as follows.
\begin{definition} \label{def:adaptive_lg_complexity}
    Let $\G$ be an adaptive learning graph for $f$.
    If $\F \subseteq \E$ is a stage of $\G$, for $x, y\in \cD$, we define the negative and positive complexity of 
        $\F$ respectively as
    \begin{align*}
        C_0(\F, x) := \sum_{e\in \F}l(e)w_x(e), & \quad C_0(\F) := \max_{x\in f^{-1}(0)} C_0(\F, x)\\
        C_1(\F, y) := \sum_{e\in \F}l(e)\frac{p_y(e)^2}{w_y(e)}, & \quad C_1(\F) := \max_{y\in f^{-1}(1)} C_1(\F, y)
    \end{align*}
    The \textit{adaptive learning graph complexity} of $\G$ is $\LG^{adp}(\G) := \sqrt{C_0(\E)C_1(\E)}$.
    The \textit{adaptive learning graph complexity} $\LG^{adp}(f)$ of $f$ is the minimum complexity of an 
        adaptive learning graph for $f$.
\end{definition}

Observe that \defn{learning_graph} is a special case of definition 
    \defn{adaptive_learning_graph}, so $\LG^{adp}(f) \le \LG(f)$.
There is also a dual adversary reduction for adaptive learning graphs \cite{CLM19}, and we get a similar 
    upper bound result.
\begin{theorem} \label{thm:adaptive_learning_graph_upper}
    For any (possibly partial) Boolean function $f$, $\Q(f) = O(\LG^{adp}(f))$.
\end{theorem}

An example of this framework is a learning graph version of quantum walks on Johnson graph \cite{CLM19}. The stages in this learning graph are symmetric.
\begin{lemma}[Learning graph for Johnson walk, \cite{CLM19}]
\label{lem:learning_graph_johnson_walk}
    Let $\ell \le k = o(n)$.
    For each $A \in \binom{[n]}{k}$, 
    let $f_A\colon \{0, 1\}^N \to \{0, 1\}$ be a Boolean function.
    Define $f = \bigvee_{A \in \S_k([n])}f_A$. This is a function on $N$ bits.

    Let the data structure $D$ be a monotone mapping
    (preserving inclusion under subsets)
    from $\P([n])$ to $\P([N])$ 
        such that for every 1-input $x$ of $f$, there is some
        $I_x \in\binom{[n]}{\ell}$
        such that $D(I_x)$ is a 1-certificate of $x$ with respect to $f$.
    For $\lambda$ a partial assignment on $N$ bits, 
    let $f_{A, \lambda}$ be the Boolean function
    which outputs $1$ on $z\in\B^N$ if both $f_A(z)=1$
    and $z_{D(A)}=\lambda$.
    We have $f_A = \bigvee_{\lambda}f_{A, \lambda}$ where $\lambda$ ranges over all partial assignments on $N$ bits. 
    Suppose $\G_{A, \lambda}$ is a learning graph for $f_{A, \lambda}$.

    Let $\bS, \bU, \bC > 0$ be values such that for every $x \in f^{-1}(0)$, we have 
    \begin{align}
        \Exp_{A \in \binom{[n]}{k-\ell}}|D(A)|^2 & \le \bS^2,\\
        \Exp_{\substack{A \in \binom{[n]}{i}\\ v \in [n] \setminus A}}|D(A \cup \{v\}) \setminus D(A)|^2
            & \le \bU^2, \text{ for } k - \ell \le i < k\\
        \Exp_{A\in \binom{[n]}{k}}\squareb{C_0(\G_{A, x_{D(A)}}, x)\cdot C_1(\G_{A, x_{D(A)}})} & \le \bC^2.
    \end{align}
    Then there is an adaptive learning graph $\G$ for $f$ such that for every $x \in f^{-1}(0), y \in f^{-1}(1)$,
    \begin{equation*}
        C_0(\G, x) = O\squareb{\bS^2 + \roundb{\frac{n}{k}}^\ell\roundb{k\cdot \bU^2 + \bC^2}} \quad \text{ and } \quad
        C_1(\G, y) \le 1.
    \end{equation*}
\end{lemma}
Taking a square root of the 0-complexity of $\G$ gives the same complexity bound of the original quantum walk. In other words, this lemma is saying
that if the ``checking'' part of a quantum walk on a Johnson graph
can be implemented by learning graphs $\G_{A,\lambda}$,
a quantum walk on a Johnson graph which computes $f$ using the data
structure $D$ can also be implemented by an adaptive
learning graph with the same cost.

In the rest of this paper, we will use ``learning graph'' to refer
to an adaptive learning graph.


\section{Reductions for simplex finding}

We study the problem of simplex finding in a hypergraph;
this is a generalization of triangle finding in a graph.
We start by reviewing some trivial upper and lower bounds
for the quantum query complexity of simplex finding.
Then we give a more interesting reduction between simplex
finding for hypergraphs of different rank.

\subsection{Basic properties of simplex finding}

We define the simplex-finding problem
$\SF_{n,r}\colon \B^{\binom{n}{r}}\to\B$ as follows.
The input string is interpreted as a function $x\colon\binom{[n]}{r}\to\B$,
where $x(S)=1$ means that $S\subseteq\binom{[n]}{r}$ is a present
hyperedge in the $r$-uniform hypergraph defined by $x$.
The function $\SF_{n,r}(x)$ evaluates to $1$ if and only if
there exists a simplex
in this hypergraph; that is, if and only if there exists a set of vertices
$V\in\binom{[n]}{r+1}$ such that $x(V\setminus\{v\})=1$
for each $v\in V$.

We note that $\SF_{n,2}$ is triangle finding and $\SF_{n,3}$ is tetrahedron
finding. We also note that $\SF_{n,1}$
asks if the Hamming weight of an input string in $\B^n$ is at least $2$;
hence $\SF_{n,1}$ can be thought of as a variant of Grover search.
$\SF_{n,0}$ is the identity function from $\B$ to $\B$.

The following easy query complexity bounds hold for simplex finding.

\begin{lemma}\label{lem:SFbounds}
For any constant rank $r$, we have $\displaystyle 
    \Q(\SF_{n,r})=O(n^{(r+1)/2}) \text{ and } \Q(\SF_{n,r})=\Omega(n^{r/2})$.
\end{lemma}

\begin{proof}
Let $G = (V, E)$ be an $r$-uniform hypergraph.
Given a set of vertices
$e=\{v_1, \dots, v_r\}$, we use $e_{\hat{i}}$ to denote the subset
$\{v_1,v_2, \dots, v_{i-1},v_{i+1},\dots, v_r\} \in \binom{V}{r-1}$.
We write $e_{S}$ for $S \subseteq [r]$ to denote the subset 
    $\{v_i : i \in S\}\in \binom{V}{|S|}$.
For vertex $u\in V$, we use $ue$ to abbreviate the subset $\{u\}\cup e$.

Since a 1-certificate of $\SF_{n,r}$ is given by finding an 
$(r + 1)$-sized subset of vertices and checking all $r + 1 = O(1)$ 
possible $r$-edges formed by these vertices, we can detect an $r$-simplex 
in $G$ by Grover searching over sets of $r+1$ vertices,
and for each one checking all $r+1$ hyperedges formed by removing
a single vertex from this set. Implementing the inner search
with Grover search as well, this can be done using
$O\roundb{\sqrt{\binom{n}{r + 1}(r+1)}} = O\roundb{n^{(r+1)/2}}$
quantum queries.

For the lower bound, we suppose $V = \{v_0, v_1, \dots, v_{n-1}\}$.
Impose the following promise on the input:
for each subset of indices $\displaystyle S = \curlyb{i_1, \dots, i_{r-1}}
    \in \binom{[n-1]}{r-1}$, we are promised that 
    $\{v_0,v_{i_1},\dots ,v_{i_{r-1}}\} \in E$.
Under this promise, to find an $r$-simplex in $G$, it is necessary 
    and sufficient to find an $r$-edge
    among the vertices $\{v_1, v_2, \dots, v_{n-1}\}$.
This is equivalent to unordered search for a $1$ in the function $x$,
    restricted to the inputs $\binom{\{v_1, \dots, v_{n-1}\}}{r}$ of the function.
This search requires $\Omega(\sqrt{\binom{n-1}{r}}) = \Omega((n/r)^{r/2})$ 
    queries due to lower bound on unordered search \cite{BBBV97}.
Since adding a promise to the $r$-hypergraphs can only reduce query complexity, we obtain 
    $\displaystyle Q(\SF_{n,r}) = \Omega\roundb{(n/r)^{r/2}}.$
\end{proof}
The main objective of studying simplex finding problems is to find the exponent 
    $\displaystyle \frac{r}{2} \le a_r \le \frac{r + 1}{2}$ for which
    $Q(\SF_{n,r}) \in O(n^a\cdot g(n))\cap \Omega(n^a / g(n))$ 
    for some subpolynomial factor $g(n)$,
    or at least reduce the range we have on this exponent $a_r$.

\subsection{From high rank to low rank}

To this date, the trivial $\Omega(n^{r/2})$ query complexity in 
    \lem{SFbounds} is still the best known lower bound
    for simplex finding in every rank $r$.
However, we are able to uncover interesting relationships 
    connecting the query complexity of simplex finding of different ranks.
Intuitively, a tetrahedron should have more structural information than 
    a triangle, and therefore should be more difficult to find;
    this might suggest that a nontrivial lower bound for triangle finding
    should give rise to a nontrivial lower bound for tetrahedron finding.
However, this is not immediately the case, because what counts as a ``trivial''
    lower bound for tetrahedron finding is a larger query complexity than what
    counts as a trivial lower bound for triangle finding!
    
We show a stronger reduction: the ability to solve tetrahedron finding
    can be leveraged to solve not just triangle finding, but a search
    over multiple instances of triangle finding.
\begin{theorem} \label{thm:reductionlb}
    For any rank $r \ge 2$, we have $\, \Q\roundb{\SF_{2n, r+1}}\, 
        = \, \Omega\roundb{\sqrt{n} \cdot \Q(\SF_{n, r})}$.
\end{theorem}

Since the growth rate of $\Q(\SF_{n,r})$ is polynomial in $n$,
this theorem implies \thm{reduction}.

Since the best known quantum query upper bound for triangle
    finding is $O(n^{1.25})$, this result also provides an approach to 
    improve triangle finding algorithm by finding query-efficient 
    algorithm for finding higher-rank simplex in hypergraphs.
In particular, the following corollary is a direct consequence of
    \thm{reduction}.
\begin{corollary} \label{cor:ubreduction}
    If there is a quantum algorithm solving \emph{Tetrahedron} with $o(n^{1.75})$ queries, 
        then there is a quantum algorithm solving \emph{Triangle} with $o(n^{1.25})$ queries.
\end{corollary}

It remains to prove \thm{reductionlb}, which we will 
    do with a randomized reduction.

\begin{proof}[Proof of \thm{reductionlb}]
        Consider two disjoint sets of vertices $A, B$ where 
        $|A| = |B| = n$.
    For every vertex $v \in A$, assume there is an associated 
        $r$-uniform hypergraph $G_v$ on vertex set $B$.
    Let $E_v$ be the set of $r$-edges of $G_v$ and suppose that 
        $E_v$ can be accessed with an oracle query to the pair 
        $(v, e) \in A \times \binom{B}{r}$.
    Then the problem of finding an $r$-simplex in any of the $G_v$ 
        is equivalent to the Boolean function 
        $\OR_n \circ \SF_{n, r}^n$.
    By \thm{composition_bound}, the quantum query complexity of this 
        problem is 
    \[\Q(\OR_n \circ \SF_{n, r}^n) = \Theta\roundb{\Q(OR_n)
        \Q(\SF_{n, r})} = \Theta\roundb{\sqrt{n} \cdot 
        \Q(\SF_{n, r})}.\]

    Let $\SF_{A, B, r}$ denote the $r$-simplex finding problem on 
        $r$-hypergraph $G'$ with the promise that $G'$ 
        has vertex set $A \cup B$, no $r$-edge in $G'$ has more 
        than 1 vertex in $A$, and $G'_B$ is a complete $r$-partite
        hypergraph with $r$-partition $B_1, \dots, B_r$ of equal size.
    We call the $r$-edges with exactly one vertex in $A$ type 1
        hyperedges and the $r$-edges in $G'_B$ the type 2 hyperedges.
    Note that type 1 and type 2 hyperedges are disjoint.
    Furthermore, the $r$-partition $B_1, \dots, B_r$ is known
    and therefore deciding type 2 hyperedges doesn't 
        cost any queries.

    Given an instance of the $\OR_n \circ \SF_{n, r}^n$ problem 
        described above, we will ``increase the rank'' and construct 
        an $(r+1)$-uniform hypergraph $G$ with randomization.
    Let the vertex set of $G$ be $V = A\cup B$ and define the 
    hyperedges in $G$ according to the two types $E = E_1 \cup E_2$.
    Let $E_1 := \{e\cup \{v\} : v\in A, e \in E_v\}$ be the set of  
        $(r+1)$-edges of $G$ constructed from $r$-edges in $G_v$.
    To construct $E_2$, we will uniformly randomly pick an 
        $(r+1)$-partition $B_1, B_2, \dots, B_{r+1}$ of $B$ 
        such that $|B_1| = |B_2| = \dots |B_{r+1}| = \frac{n}{r + 1}$.
    Then define $E_2$ as $K_{B_1, B_2, \dots, B_{r+1}}$, 
        the $(r+1)$-edges of the complete $(r+1)$-partite graph.
    Note that the $(r+1)$-hypergraph $G$ we constructed is an 
        instance of the $\SF_{A,B,r+1}$ problem.
    This construction is depicted in     
        \fig{rank_increase_reduction}.
    Moreover, if $G'$ is an $(r+1)$-hypergraph with the promise of 
        the $\SF_{A,B,r+1}$ problem, then for every $v \in A$, 
        we can define an $r$-hypergraph $H_v$ on vertex set $B$ 
        such that $v_1v_2\dots v_r$ is an $r$-edge of $H_v$ if and
        only if $vv_1v_2\dots v_r$ is a type 1 hyperedge of $G'_v$.
    Note that $(v, H_v)_{v \in A}$ is an instance of the $\OR_n \circ 
        \SF_{n, r}^n$ problem and $G'$ can only be obtained from 
        $(v, H_v)_{v \in A}$ via the rank-increase construction.

        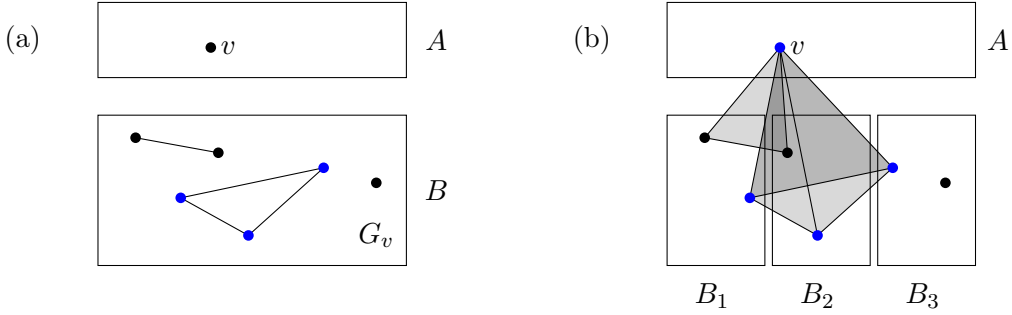
\begin{figure}
            \hfill
            \begin{tikzpicture}[scale=1] 
                \node (a) at (-1, 3) {(a)};
                \draw (0, 0) rectangle (4.1,2);
                \node (B) at (4.5, 1) {$B$};
                \node (B1) at (0.6, -0.4) {\phantom{$B_1$}};
                \draw (0, 2.5) rectangle (4.1, 3.5);
                \node (A) at (4.5, 3) {$A$};
                \node (Gv) at (3.7, 0.4) {$G_v$};
                
                \coordinate [label=right:$v$] (v) at (1.5, 2.9);
                \coordinate (1) at (1.1, 0.9);
                \coordinate (2) at (2, 0.4);
                \coordinate (3) at (3, 1.3);
                \coordinate (4) at (0.5, 1.7);
                \coordinate (5) at (1.6, 1.5);
                \coordinate (6) at (3.7, 1.1);

                \foreach \from/\to in {1/2, 2/3, 1/3, 4/5}
                    \draw (\from) -- (\to);
                
                \foreach \x in {v,4,5,6}
                    \fill[black] (\x) circle (2pt);
                \foreach \x in {1,2,3}
                    \fill[blue] (\x) circle (2pt);
            \end{tikzpicture} 
            \hfill
            \begin{tikzpicture}[scale=1] 
                \node (b) at (-1, 3) {(b)};
                \coordinate [label=right:$v$] (v) at (1.5, 2.9);
                \coordinate (1) at (1.1, 0.9);
                \coordinate (2) at (2, 0.4);
                \coordinate (3) at (3, 1.3);
                \coordinate (4) at (0.5, 1.7);
                \coordinate (5) at (1.6, 1.5);
                \coordinate (6) at (3.7, 1.1);

                \begin{scope}[transparency group]
                \begin{scope}[blend mode=multiply]
                    \foreach \a/\b in {1/2, 2/3, 1/3, 4/5}
                        \fill[gray!30] (v) -- (\a) -- (\b) -- cycle;
                \end{scope}
                \end{scope}
  
                \draw (0, 0) rectangle (1.3,2);
                \draw (1.4, 0) rectangle (2.7,2);
                \draw (2.8, 0) rectangle (4.1,2);
                \node (B1) at (0.6, -0.4) {$B_1$};
                \node (B2) at (2, -0.4) {$B_2$};
                \node (B3) at (3.4, -0.4) {$B_3$};
                \draw (0, 2.5) rectangle (4.1, 3.5);
                \node (A) at (4.4, 3) {$A$};
                
                \foreach \from/\to in {1/2, 2/3, 1/3, 4/5, v/1, v/2, v/3, v/4, v/5}
                    \draw (\from) -- (\to);
                
                \foreach \x in {4,5,6}
                    \fill[black] (\x) circle (2pt);
                \foreach \x in {v,1,2,3}
                    \fill[blue] (\x) circle (2pt);
            \end{tikzpicture} 
            \hspace{1.5cm}
            
            \caption[
                Example of quantum query lower bound by rank reduction.
            ]{
                An example of the rank lower bound reduction when $r = 2$.
                This example shows that $\Q\big(\SF_{2n,3} \big) = 
                    \Omega\big(\sqrt{n} \cdot \Q(\SF_{n,2}) \, \big)$.
                In particular, any nontrivial lower bound of the triangle finding problem implies a 
                    nontrivial lower bound for the tetrahedron finding problem.
                (a) depicts an instance of the $\OR_n \circ \SF_{n,2}$ problem with
                    $G_v$ shown for a particular $v \in A$.
                The blue vertices form a triangle.
                (b) depicts an instance of the $\SF_{A,B,3}$ problem obtained from (a) by the 
                    rank increase construction.
                The gray-shaded triangles are the type 1 $3$-hyperedges.
                The 3-partition of $B$ is randomly chosen and forms a complete $3$-hypergraph, 
                    so the blue vertices form a tetrahedron.      
            } \label{fig:rank_increase_reduction}
        \end{figure}

    Suppose there are vertices $v_1, v_2, \dots, v_{r+1} \in \S_{r+1}(B)$ that form an $r$-simplex in $G_v$.
    Then for each $i \in [r+1]$, $e_{\hat{i}} \in E_v$ and $\{v\}\cup e_{\hat{i}}$ are type 1 hyperedges of $G$.
    Let $P$ be the event that each of these $r+1$ vertices fall in a distinct partition of $B$. Then
    \begin{equation*}
            \Pr(P) = \Pr_{B = B_1 \cup \dots \cup B_{r+1}}\squareb{\exists_{\pi \in S_{r+1}}\, \forall_{i\in [r+1]}\,
                v_i \in B_{\pi(i)} } = \prod_{i = 1}^r\frac{r+1-i}{r+1}\cdot\frac{n}{n-i}
    \end{equation*}
        where $S_{r+1}$ is the symmetric group of $r+1$ vertices.
    Note that $\Pr(P)$ is a constant when $r$ is a constant.
    In the event of $P$, $v_1v_2\dots v_{r+1}$ becomes a type 2 hyperedge of $G$.
    Together with the type 1 hyperedges $\{v\}\cup e_{\hat{i}}$ in $G$, the vertices $\{v, v_1, v_2, \dots, v_{r+1}\}$ 
        form an $(r+1)$-simplex of $G$.

    Suppose $G'$ is an instance of $\SF_{A,B,r+1}$ where $G'$ 
        is obtained from $(v, H_v)_{v \in A}$ via the 
        rank-increase construction.
    If $\displaystyle u, u_1, \dots, u_{r+1} \in 
        A \times \binom{B}{r+1}$ is a set of vertices that formed 
        an $(r+1)$-simplex in $G'$, then $u_1, \dots, u_{r+1}$ 
        must be an $r$-simplex in $H_u$.
    Moreover, every type 1 hyperedge query in $\SF_{A,B,r+1}$ is
        equivalent to a query of the form 
        $(u, e)$ in $\OR_n \circ \SF_{n, r}^n$.
    Therefore, we can solve the $\OR_n \circ \SF_{n, r}^n$ problem 
        by solving an $\SF_{A,B,r+1}$ problem using
        the same amount of quantum queries.
    Note that $\SF_{A,B,r+1}$ is a promise problem of 
        $\SF_{2n,r+1}$.
    Since the randomized reduction success with probability at least 
        $\Pr(P) = \Theta(1)$, we observe that \begin{equation*}
        \Q\big(\SF_{2n,r+1} \big) = \Omega\squareb{\Q(\SF_{A, B,r+1})}
        = \Omega\Big[\Q(\OR_n \circ \SF_{n, r}^n)\Big]
        = \Omega\big(\sqrt{n} \cdot \Q(\SF_{n, r}) \, \big).\qedhere
    \end{equation*}
\end{proof}


\section{Converting nested quantum walks to adaptive learning graphs}

In this section, we explain how to formulate the nested quantum walk algorithm 
    in an adaptive learning graph.
In the next section, we will use this newly developed framework to find a
    nontrivial algorithm for the 4-simplex finding problem.

Let's start by reviewing nested quantum walks, which were first introduced by
Jeffery, Kothari, and Magniez \cite{JKM13}. Quantum walks are nested when the checking
procedure of one quantum walk is another quantum walk.
An $r$-level nested quantum walk uses a state tuple $(A_1, A_2, \dots, A_r)$ where $A_i$ is the state of the $i^{th}$ level quantum walk.
However, instead of keeping a separate data structure $D(A_i)$ at each level, it keeps track of a data structure
in a global quantum state $\ket{A_1, A_2, \dots, A_r, D(A_1, \dots, A_r)}$.
This allows us to push the setup cost of the quantum walk in every level to the beginning of the computation.

We are interested in the case where each level of the nested quantum walk is just a symmetric walk on a Johnson graph (this is the usual case for nested quantum walks). In that setting, 
we show that we can convert such a nested quantum walk into an adaptive learning graph.
The learning graph framework is additionally easier to analyze;
in the next section, we utilize this framework to find a non-trivial algorithm for
$4$-simplex finding.
Our approach to the conversion extends \lem{learning_graph_johnson_walk};
however, we need to make a few important modifications.

\subsection{Configuration Packages}

Our objective is to formulate a learning graph for nested 
    quantum walks on Johnson graphs $J([n_i], k_i)$.
To do that, we need the state space of each quantum walk in the hierarchy to be
dependent on the state of the previous (outer) walks. 
To this end, we label the L-vertices of our learning graph by ordered partial subsets instead of subsets.
This allows us to refer to a particular element in the state by its position.

For a set $X$ and an integer $k$, define the set of ordered partial subsets of size $k$ as 
    \begin{equation}
        \P(X, k) := \Big\{(x_1, \dots, x_k) \in (X \cup \{\star\})^k: \text{ for } i \neq j \in [k], x_i = x_j 
        \implies x_i = \star \Big\}.
    \end{equation}
We also define the set of ordered subsets of size $k$ as \begin{equation}
    \P(X, =k) := \Big\{(x_1, \dots, x_k) \in X^k: \text{ for } i \neq j \in [k], x_i \neq x_j \Big\}.
\end{equation}
Here the $\star$ symbol is a placeholder that refers to an element of the subset not yet determined.
We can treat $A\in \P(X, k)$ as a set by ignoring the star symbols and treat the elements in $A$ as unordered;
    this allows us to generalize membership and set difference to $A$.
We define the size of $A$ (denoted by $|A|$) as the number of non-star elements in $A$.
If $|A| < k$, we say $A$ is partially filled.
If $A$ is partially filled, then for $v \not\in A$, we use the notation $A\cup \{v\}$ to randomly replace 
    a $\star$ symbol in $A$ by $v$. 
For $A, B \in \P(X, k)$, we write $A \subseteq B$ if for every $i\in [k]$ where $A_i \neq \star$, we have $A_i = B_i$.

The $i^{\text{th}}$ level of the nested walk is labeled by $\P([n_i], k_i)$.
The certificate of the nested quantum walk is given by a sequence
    $I_y = (I_{y, 1}, \dots, I_{y, r})$ where each 
    $I_{y, i}\in\binom{[n_i]}{\ell_i}$. 
Define the set \begin{equation}
    \overline{I_{y, i}} := \curlyb{A'_i \in \P([n_i], k_i) : |A'_i| = k_i - \ell_i,\  A'_i\cap I_{y,i} = \emptyset}.
\end{equation}
We say that a state $A'_i \in \P([n_i], k_i)$ \textit{avoids} the certificate       $I_{y, i}$ if $A'_i \in \overline{I_{y, i}}$.
We usually attach a prime symbol for elements of $\overline{I_{y, i}}$ and we will use these states often during 
    the setup stages of the nested quantum walk.
In the setup of the $i^{\text{th}}$ level state $A_i$, we assume we have the setup states of the earlier levels $A'_1, \dots, A'_{i-1}$,
so the certificate of the $i^{\text{th}}$ level can utilize this information;
    we further assume that 
    $I_{y, i} = I_{y, i, A'_1, \dots, A'_{i-1}}$ depends on these setup states.

When we design the flow for a learning graph of quantum walk, a valid state in the $i^{\text{th}}$ level should have 
    the form $A'_i \cup I_{y, i}$.
However, there may be special circumstances we want to avoid, even when $A_i$ contains the certificate $I_{y,i}$. 
For this purpose, we define an availability function $C$ such that 
    $C(A'_1, \dots, A'_{i-1}) \subseteq \overline{I_{y, i}}$.
We design the learning graph such that $A_i$ is valid if and only if $A_i = A'_i \cup I_{y, i}$ for some 
    $A'_i \in C(A'_1, \dots, A'_{i-1})$.
In our applications, the proportion of unavailable states avoiding $I_{y, i}$ is small.
That is, for some function $\alpha = o(1)$, a function of $n$ and fixed setup states $A'_1, \dots, A'_{i-1}$, 
    we assume that \begin{equation} \label{eqn:small_ava_set}
    \Pr_{A'_i \in \overline{I_{y, i}}} [A'_i \not\in C(A'_1, \dots, A'_{i-1})] \le \alpha.
\end{equation}
In this case, we call $C(A'_1, \dots, A'_{i-1})$ an \textit{$\alpha$-subset} of $\overline{I_{y, i}}$.
If $C(A'_1, \dots, A'_{i-1}) = \overline{I_{y, i}}$, we say that $C$ is \textit{trivial} for this level.

The data structure associated with the nested quantum walk is given by a monotone function 
    $D$ mapping from $\prod_{i=1}^r \P([n_i], k_i)$ to $\P([N])$.
This data structure is kept at the earliest level of the nested quantum walk so that all levels have access to 
    the data structure.

Finally, let's formalize these ideas by grouping all the sets and parameters defined above 
    into a configuration package used to define the learning graph of a nested Johnson walk.
\begin{definition} \label{def:lg_configuration}
    For each $i\in [r]$, let $0 < \ell_i \le k_i = o(n_i)$ be integer parameter where $\ell_i$ is a constant.
    Let $f_{A_1, \dots, A_r}: \{0, 1\}^N \to \{0, 1\}$ be Boolean functions and suppose \begin{equation}
        f = \bigvee_{A_i \in \P([n_i], =k_i),\ i\in [r]}f_{A_1, \dots, A_r}
    \end{equation} is the function we are trying to compute.
    Define the \textit{configuration} of a nested Johnson walk learning graph computing $f$
        as the tuple \begin{multline} \label{eqn:lg_config_tuple}
            \Big(\big\{f_{A_1, \dots, A_r}\big\}_{A_i \in \P([n_i], =k_i), i \in [r]}, \ 
                \big\{(n_i, k_i, \ell_i)\big\}_{i\in [r]}, \ 
                \big\{I_y\big\}_{y \in f^{-1}(1)},\  C\  , \alpha,\  D,\\
        \big\{\G_{A_1, \dots, A_r, \lambda}\big\}_{\substack{A_i \in \P([n_i], =k_i),\  i \in [r],\ \\
                \lambda \text{ partial assignment on } D(A_1, \dots, A_r)}} \Big).
        \end{multline}
    Here, $r$ is the number of levels in the nesting structure.
    $\alpha(n) = o(1)$ is a function of $n$. 
    The variables $\{I_y\}, C, D$ respectively denote the sequence of certificates, the availability function, 
        and the data structure explained in this subsection.
    For each $\lambda$ a partial assignment on $D(A_1, \dots, A_r)$, let $f_{A_1, \dots, A_r, \lambda}$ be the 
        partial Boolean function $f_{A_1, \dots, A_r}$ restricted to inputs $z \in \{0, 1\}^N$ where 
        $z_{D(A_1, \dots, A_r)} = \lambda$.
    Then, we have $f_{A_1, \dots, A_r} = \bigvee_{\lambda} f_{A_1, \dots, A_r, \lambda}$.
    Furthermore, each $\G_{A_1, \dots, A_r, \lambda}$ is a learning graph for $f_{A_1, \dots, A_r, \lambda}$.
\end{definition}

The following conditions on the configuration make sure the sequence of certificates can depend on 
    previous setup states, as we explained above.
\begin{definition} \label{def:lg_admisibile_config}
    We say that the configuration in equation (\ref{eqn:lg_config_tuple}) is \emph{admissible} if for every 
        $y \in f^{-1}(1)$, there is $I_{y, 1} \in \binom{[n_1]}{\ell_1}$ and an $\alpha$-subset 
        $C() \subseteq \overline{I_{y, 1}}$, such that for every $A'_1 \in C()$, 
        there is $I_{y, 2}\in \binom{[n_2]}{\ell_2}$ and an $\alpha$-subset $C(A'_1) \subseteq \overline{I_{y, 2}}$, 
        such that for every $A'_2 \in C(A'_1)$, $\dots$, there is $I_{y, r}\in \binom{[n_r]}{\ell_r}$ and 
        an $\alpha$-subset $C(A'_1, \dots, A'_{r-1}) \subseteq \overline{I_{y, r}}$, such that for every 
        $A'_r \in C(A'_1, \dots, A'_{r-1})$, we have \begin{equation}
            f_{A'_1\cup I_{y, 1}, \dots, A'_r\cup I_{y, r}}(y) = 1.
        \end{equation}
\end{definition}

\subsection[Alpha Symmetric Stage]{$\alpha$-symmetric Stage}

Here, we investigate what happens when we drop an 
    $\alpha$-fraction of valid L-vertices from its flows.
\begin{definition} \label{def:alpha_symmetric_stage}
    Suppose $\F$ is a learning graph stage with starting L-vertices $V_i$ and ending L-vertices $V_j$.
    Define $c := |V_i|$, $e = |V_j|$ and set $s$ as a constant.
    Let $V_{i, y}, V_{j, y}$ be the set of vertices in $V_i, V_j$ respectively which receive positive flow from $p_y$.
    For $v \in V_{i, y} \cup V_{j, y}$, let $p'_{v, y}$ denote the value of the positive flow through vertex $v$. 
    We say $\F$ is $\alpha$-\emph{symmetric with constant $s$} if it can be obtained via the following operations:
    \begin{enumerate}
        \item Suppose we have a symmetric stage $\F'$ in \defn{symmetric_stage} with 
            parameters $c, c', d, d', e, e'$.
        We let $\F$ inherit the L-vertices and L-edges of $\F'$. 
        It remains to define the flow of $\F$.
        
        \item For any $y\in f^{-1}(1)$, there is a set $V'_{i, y} \subseteq V_i$ of beginning vertices 
            receiving positive flow from $p_y(\F')$ where $|V'_{i, y}| = c'$.
        The set $V_{i, y}$ is obtained by removing a small fraction of L-vertices from $V'_{i, y}$ such that
            $(1-\alpha)^s c' \le |V_{i, y}| \le c'$.

        \item Let $V'_{j, y} \subseteq  V_j$ be the set of ending vertices receiving positive flow from 
            $p_y(\F')$ where $|V'_{j, y}| = e'$.
        The set $V_{j, y}$ is obtained by removing L-vertices from $V'_{j, y}$ such that
            \begin{equation} \label{eqn:end_vertex_conds}
                (1-\alpha)^{s + 1} e' \le |V_{j, y}| \le e' \quad\text{and}\quad 
                \frac{|v_+ \cap V_{j, y}|}{|v_+ \cap V'_{j, y}|} \ge 1 - \alpha \quad\text{for every } v \in V_{i, y}
            \end{equation}
            where $v_+$ denotes the set of out-neighbours of a vertex $v$.
        This ensures the subset to be removed from $V_j$ doesn't target any $v \in V_{i, y}$.

        \item The flows in $\F$ inherit the flows in $\F'$ with a few alternations. 
        The flow of an L-edge is reduced to zero if the L-edge doesn't lie in $V_{i, y} \times V_{j, y}$.
        To compensate for the total value loss, the flows of the L-edges in $V_{i, y} \times V_{j, y}$ 
            are scaled by a factor of at most $\frac{1}{(1 - \alpha)^{s+1}}$.
        This ensures that for $v \in V_{i, y}, w \in V_{j, y}$, we have \begin{equation}
            \frac{1}{c'} \le |p'_{v, y}| \le \frac{1}{(1 - \alpha)^s c'} \quad\text{ and }\quad 
            \frac{1}{e'} \le |p'_{w, y}| \le \frac{1}{(1 - \alpha)^{s+1} e'}.
        \end{equation}
    \end{enumerate}
\end{definition}

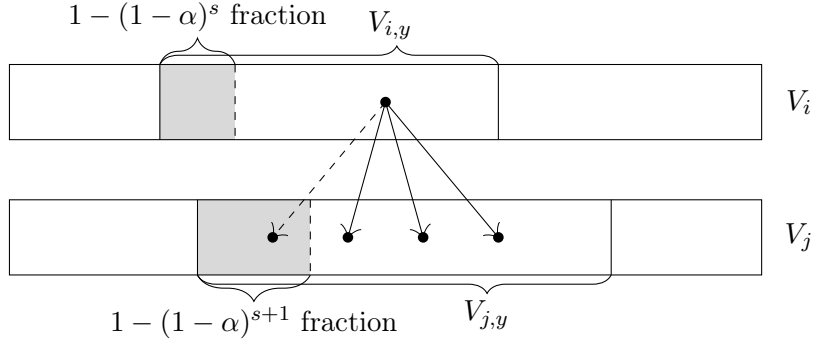
\begin{figure}
        \centering
        \begin{tikzpicture}[scale=1]

            \coordinate (1) at (2, 2.8);
            \coordinate (2) at (2, 1.8);
            \coordinate (3) at (6.5, 2.8);
            \coordinate (4) at (6.5, 1.8);
            \coordinate (5) at (2.5, 1);
            \coordinate (6) at (2.5, 0);
            \coordinate (7) at (8, 1);
            \coordinate (8) at (8, 0);

            \coordinate (9) at (3, 2.8);
            \coordinate (10) at (3, 1.8);
            \coordinate (11) at (4, 1);
            \coordinate (12) at (4, 0);

            \fill[gray!30] (1) rectangle (10);
            \fill[gray!30] (5) rectangle (12);

            \draw (0, 0) rectangle (10,1);
            \node (B) at (10.5, 0.5) {$V_j$};
            \draw (0, 1.8) rectangle (10, 2.8);
            \node (A) at (10.5, 2.3) {$V_i$};
            
            \foreach \from/\to in {1/2, 3/4, 5/6, 7/8}
                \draw (\from) -- (\to);
            \foreach \from/\to in {9/10, 11/12}
                \draw[dashed] (\from) -- (\to);

            \coordinate (v) at (5, 2.3);
            \coordinate (v1) at (3.5, 0.5);
            \coordinate (v2) at (4.5, 0.5);
            \coordinate (v3) at (5.5, 0.5);
            \coordinate (v4) at (6.5, 0.5);
            
            \foreach \x in {v, v1, v2, v3, v4}
                \fill[black] (\x) circle (2pt);
            \foreach \from/\to in {v/v2, v/v3, v/v4}
                \draw[-{Classical TikZ Rightarrow[length=1.5mm]}] (\from) -- (\to);
            \foreach \from/\to in {v/v1}
                \draw[-{Classical TikZ Rightarrow[length=1.5mm]}, dashed] (\from) -- (\to);

            \draw [decorate, decoration={brace, amplitude=7pt, raise=0ex, aspect=0.7}]
                (1) -- (3) node[midway, xshift=23pt, yshift=15pt]{$V_{i, y}$};
            \draw [decorate, decoration={brace, amplitude=7pt, mirror, raise=0ex, aspect=0.7}]
                (6) -- (8) node[midway, xshift=31pt, yshift=-15pt]{$V_{j, y}$};
            \draw [decorate, decoration={brace, amplitude=10pt, raise=0ex}]
                (1) -- (9) node[midway, yshift=18pt]{$1 - (1 - \alpha)^s$ fraction};
            \draw [decorate, decoration={brace, amplitude=10pt, mirror, raise=0ex}]
                (6) -- (12) node[midway, yshift=-18pt]{$1 - (1 - \alpha)^{s+1}$ fraction};
            
        \end{tikzpicture}
        
        \caption[
            Example of an $\alpha$-symmetric stage.
        ]{
            Example of an $\alpha$-symmetric stage.
            $V_i$ is the set of beginning vertices, $V_j$ is the set of ending vertices.
            The arrows mark the positive flows from the original symmetric stage.
            The gray areas are vertices later marked as ``unavailable''.
            As a result, the dash L-edge is removed from the stage and the value of its flow 
                is redistributed to the remaining available L-edges.
        } \label{fig:alpha_symmetric_stage}
    \end{figure}

Note that if stage $s$ of a learning graph is $\alpha$-symmetric with constant $s$, then 
    operation 2 holds for stage $s + 1$ since operation 3 holds for stage $s$.
Thus, we can design stage $s + 1$ as an $\alpha$-symmetric stage with constant $s + 1$.
We will design learning graphs consisting of sequential $\alpha$-symmetric stages.
Provided the number of stages is constant, the constant $s$ is irrelevant 
    to the overall complexity of the learning graph. 
An example of $\alpha$-symmetric stage is presented in \fig{alpha_symmetric_stage}.

The following lemma analyzes the complexity of one $\alpha$-symmetric stage.
\begin{lemma} \label{lem:alpha_symmetric_stage}
    Let $\F$ be an $\alpha$-symmetric stage of $\G$ with constant $s$ and let $\F'$ be its underlying symmetric stage. 
    Let $T := cd/c'd'$.
    For every $y \in f^{-1}(1)$, if $L$ is the average length of the L-edges
        receiving positive flow in $\F'$, then the L-edges in $\F$ can be weighted so that
    $$ C_0(\F) \le T \cdot L^2 \quad\text{ and }\quad C_1(\F, y) \le (1-\alpha)^{-2(s + 1)} = O(1).$$
\end{lemma}

\begin{proof}
    By \lem{symmetric_stage}, we can assign weights $w(e)$ to symmetric stage $\F'$ so that $\F'$ 
        has 0-complexity $\le T\cdot L^2$ and 1-complexity $\le 1$.
    Now if the same weight assignment is to be applied to $\F$, according to \defn{lg_complexity}, 
        the 0-complexity stays the same.
    The 1-complexity may differ in the following ways:
    \begin{itemize}
        \item The 1-complexity of the $\F$ may reduce because terms in equation (\ref{eqn:lg_complexity}) 
            corresponding to L-edges that don't belong to $V_{i, y} \times V_{j, y}$ should be removed from 
            the calculation.
        This doesn't change the 1-complexity upper bound.

        \item Due to the redistribution of flows and operation 3 from \defn{alpha_symmetric_stage},
        each L-edge has its flow scaled by a factor at most $\frac{1}{(1 - \alpha)^{s+1}}$, so every term 
            in equation (\ref{eqn:lg_complexity}) is multiplied by a factor of at most $(1-\alpha)^{-2(s+1)}$.
    \end{itemize}
    The statement of the lemma follows immediately.
\end{proof}

\subsection{Main Learning Graph Construction} \label{subsec:nested_qw_main_lg}

Now, we are ready to construct an adaptive learning graph for an $r$-level 
    nested Johnson walk in the following general-purpose lemma.
The following lemma is a formal restatement of \thm{conversion}.
We will prove this by constructing an adaptive learning graph.
\begin{lemma}[Learning Graph for Nested Johnson Walk] \label{lem:learning_graph_nested_johnson}
    Let \begin{equation*}
        \Big(\big\{f_{A_1, \dots, A_r}\big\}, \big\{(n_i, k_i, \ell_i)\big\}, \big\{I_y\big\}, C, 
        \alpha, D, \big\{\G_{A_1, \dots, A_r, \lambda}\big\}\Big)
    \end{equation*} 
    be an admissible configuration defined in \defn{lg_configuration} and \defn{lg_admisibile_config}.
    Let $\bS, \bU_1, \dots, \bU_r, \bC > 0$ be values such that for every $x \in f^{-1}(0)$, we have 
    \begin{align}
        & \qquad\qquad\qquad
            \Exp_{A'_i \sim \binom{[n_i]}{k_i - \ell_i}}
            |D(A'_1, \dots, A'_r)|^2  \le \bS^2, \\
        & \Exp_{A_i\sim \binom{[n_i]}{k_i}}\squareb{C_0(\G_{A_1, \dots, A_r, x_{D(A_1, \dots, A_r)}}, x)\cdot 
            C_1(\G_{A_1, \dots, A_r, x_{D(A_1, \dots, A_r)}})}  \le \bC^2,
    \end{align}
    and for every $i\in [r]$ and $k_i - \ell_i \le h < k_i$,
    \begin{align}
        \Exp_{\substack{A_j\sim \binom{[n_j]}{k_j}\ \forall j\le i-1,\\
                        A_j\sim \binom{[n_j]}{k_j -\ell_j}\ \forall j\ge i+1,\\ 
                        A_i\sim \binom{[n_i]}{h},\\
                        v\sim [n_i] - A_i}}
            |D(A_1, \dots, A_i \cup\{v\}, \dots, A'_r) - D(A_1, \dots, A_r)|^2 &\le \bU_i^2.
    \end{align}
    Then there is a learning graph $\G$ for $f$ such that for every $x \in f^{-1}(0), y \in f^{-1}(1)$, we have 
        $C_1(\G, y) \le 1$ and 
    \begin{equation}
        C_0(\G, x) = O\squareb{\bS^2 + \sum_{i=1}^r\roundb{\prod_{j=1}^i\roundb{\frac{n_j}{k_j}}^{\ell_j}}k_i\cdot \bU_i^2 + \roundb{\prod_{i=1}^r\roundb{\frac{n_i}{k_i}}^{\ell_i}}\bC^2}.
    \end{equation}
\end{lemma}

\begin{proof}
    We will construct a learning graph $\G$ computing $f$ consisting of the setup, update, and 
        checking stages analogous to the procedures of a nested quantum walk.
    $\G$ consists of $r + \sum_{i = 1}^r \ell_i + 1$ stages. The first $r$ stages are for setup, 
        and the last stage is for checking.
    The stages for update are labeled by lexicographically ordered pairs $(i, h)$ for $i\in [r], h\in [\ell_i]$.
    All setup and update stages in $\G$ are $\alpha$-symmetric.
    The labels of the L-vertices in $\G$ are given by $(A_1, \dots, A_r)$ where $A_i \in \P([n_i], k_i)$. 
    The root vertex is labeled by $(\emptyset, \dots, \emptyset)$ where $\emptyset\in \P([n_i],k_i)$ is represented by the tuple of stars
        $(\star, \dots, \star)$.

    We define stages using \defn{alpha_symmetric_stage}.
    For $i \in [r]$, the $i^{\text{th}}$ setup stage is given by loading $k_i - \ell_i$ elements to $A_i$.
    In this stage, we have
    \begin{align*}
        V'_{i, y} &= \curlyb{(A'_1, \dots, A'_{i - 1}) : A'_k \in \overline{I_{y, k}} \text{ for } k\in [i-1]}, \\
        V'_{j, y} &= \curlyb{(A'_1, \dots, A'_{i}) : A'_k \in \overline{I_{y, k}} \text{ for } k\in [i]}.
    \end{align*}
    Hence $c' = \prod_{j=1}^{i-1}\binom{k_i}{\ell_i}\cdot \permutate{n_i - \ell_i}{k_i - \ell_i}$.
    Counting the number of possible setup labels, the number of starting vertices is
        $c = \prod_{j=1}^{i-1}\binom{k_j}{\ell_j}\cdot \permutate{n_j}{k_j - \ell_j}$. 
    The beginning L-vertices have outdegree $d = \binom{k_i}{\ell_i}\cdot \permutate{n_i}{k_i - \ell_i}$.
    We define \begin{equation*} 
        V_{i, y} := \curlyb{(A'_1, \dots, A'_{i - 1})  \in V'_{i, y}: A'_k \in 
        C(A'_1, \dots, A'_{k-1}) \text{ for } k\in [i-1]}.
    \end{equation*}

    With the setup information $A'_1, \dots, A'_{i-1}$, the certificate $I_{y, i}$ is fixed.
    The L-edges receive positive flow if and only if no elements in $I_{y, i}$
        are loaded to $A_i$ and $A_i \in C(A'_1, \dots, A'_{i-1})$.
    Since our configuration is admissible, we have $d' = \binom{k_i}{\ell_i}\cdot \permutate{n_i - \ell_i}{k_i - \ell_i}$.
    Since $C(A'_1, \dots, A'_{i-1})$ is an $\alpha$-subset of $I_{y, i}$, equation (\ref{eqn:end_vertex_conds}) is
        satisfied with this construction.
    Since $\ell_i$ are constants and $\alpha = o(1)$, the speciality of this stage is $cd/c'd' = O(1)$.
    By \lem{alpha_symmetric_stage}, the sum of $0$-complexities of these $r$ stages is at most
        $O(\Exp\big[|D(A'_1, \dots, A'_r)|\big]^2) \le O(\bS^2).$

    For $i \in [r]$ and $h \in [\ell_i]$, stage $(i, h)$ consists of beginning L-vertices $(A_1, \dots, A_r)$
        where $|A_j| = k_j$ for $j\in [i-1]$, $|A_j| = k_j - \ell_j$ for $j \in [i+1, r]$, and $|A_i| = k_i - \ell_i + h - 1$.
    Here,
    \begin{equation*} 
        V'_{i, y} = \{(A_1, \dots, A_r) : I_{y, j} \subseteq A_j\  
        \forall_{j \le  i-1},\ 
        A_j \cap I_{y, j} = \emptyset\  \forall_{j\ge i+1},\  
        |I_{y, i}\cap A_i| = h - 1\}.
        \end{equation*}
    The L-edges of this stage load a new element to $A_i$, and an L-edge in this stage receives positive flow 
        from $p_y$ if and only if the new element loaded belongs to $I_{y, i}$.
    The corresponding parameter values for this stage are \begin{align*}
        c &= \binom{k_j}{\ell_j-h+1} \cdot \permutate{n_j}{k_j - \ell_j + h - 1} \cdot 
            \prod_{j=1}^{i-1}\permutate{n_j}{k_j} \cdot 
            \prod_{j=i+1}^{r}\binom{k_j}{\ell_j} \cdot \permutate{n_j}{k_j - \ell_j},\\ 
        c' &= \prod_{j=1}^r \binom{k_j}{\ell_j} \cdot \permutate{n_j - l_j}{k_j - \ell_j},
        \quad d = (\ell_j - h + 1)(n_i - (k_j - \ell_j + h - 1)),\\
        d' &= (\ell_j - h + 1)^2.
    \end{align*}
    The speciality of this stage is $\displaystyle \frac{cd}{c'd'} = 
        O\roundb{n_i \roundb{\frac{n_i}{k_i}}^{h-1}\prod_{j=1}^{i-1}\roundb{\frac{n_j}{k_j}}^{\ell_j}}$.
    By \lem{alpha_symmetric_stage}, the $0$-complexity of this stage is at most
        \[ T\cdot \Exp\big[|D(A_1, \dots, A_i \cup\{v\}, \dots, A_r) - D(A_1, \dots, A_r)|\big]^2 = 
        O\roundb{k_i \roundb{\frac{n_i}{k_i}}^{h}\prod_{j=1}^{i-1}\roundb{\frac{n_j}{k_j}}^{\ell_j} \cdot \bU_i^2}.\]

    The final stage of $\G$ performs the checking operation.
    For every beginning L-vertex $(A_1, \dots, A_r)$ where $|A_i| = k_i$ for all $i \in [r]$, 
        we attach the learning graph $\G_{A_1, \dots, A_r, x_{D(A_1, \dots, A_r)}}$ to this L-vertex and 
        rescale the weights of the L-edges in $\G_{A_1, \dots, A_r, x_{D(A_1, \dots, A_r)}}$ by 
        \[\lambda_{A_1, \dots, A_r} = \bigslash{ C_1(\G_{A_1, \dots, A_r, x_{D(A_1, \dots, A_r)}}) }
            { \prod_{i=1}^r\permutate{n_i - \ell_i}{k_i - \ell_i} \permutate{k_i}{\ell_i}}.\]
    There are $\prod_{i=1}^r\permutate{n_i}{k_i}$ beginning L-vertices and more than $(1-\alpha)^r 
        \prod_{i=1}^r\permutate{n_i - \ell_i}{k_i - \ell_i} \permutate{k_i}{\ell_i}$ of them receives positive flow.
    The values of the flow in these subroutine learning graphs inherit from the values of flow in the 
        original learning graph, rescaled by $\Theta\roundb{\prod_{i=1}^r\permutate{n_i - \ell_i}{k_i - \ell_i} 
        \permutate{k_i}{\ell_i}}^{-1}$.
    Our choice of rescaling ensures that the $1$-complexity of this stage is \begin{equation*} 
        \sum_{\substack{A'_i\in \S([n_i - \ell_i], k_i - \ell_i), \forall i\in [r]\\ 
            A'_i \in \overline{I_{y, i}},\  A'_i \cup I_{y, i} = A_i}} 
            \frac{C_1(\G_{A_1, \dots, A_r, x_{D(A_1, \dots, A_r)}})}
        {\Theta\roundb{\prod_{i=1}^r \roundb{\permutate{n_i - \ell_i}{k_i - \ell_i} \permutate{k_i}{\ell_i}}^2} \cdot 
            \lambda_{A_1, \dots, A_r}} = O(1).
    \end{equation*}
    The $0$-complexity of the final stage is \begin{equation*}
        \sum_{A_i\in \S([n_i], k_i), \forall i\in [r]} \lambda_{A_1, \dots, A_r}
        C_0(\G_{A_1, \dots, A_r, x_{D(A_1, \dots, A_r)}}, x) 
        = O\roundb{\prod_{i=1}^r \roundb{\frac{n_i}{k_i}}^{\ell_i}\cdot \bC^2}.
        \qedhere
    \end{equation*}
\end{proof}


\section{Quantum algorithm for 4-simplex finding} 
\label{sec:4_simplex_finding}

Before this work, there was no nontrivial (i.e.\ $o(n^{2.5})$) quantum
algorithm for simplex finding when $r=4$.
However, several nontrivial improvements have been made to
$3$-simplex finding algorithms.
Currently, the best-known algorithm for $3$-simplex finding uses $O(n^{1.883})$ 
    quantum queries \cite{LGNT16}.
It was achieved using a nested quantum walk that
    iteratively searches for vertices, 
    pairs of vertices, and the hyperedges of a $3$-simplex.
However, this algorithm doesn't use an adaptive learning graph its analysis resorts
    to analyzing quantum states during the computation.
We build on this work to obtain a $4$-simplex finding algorithm;
the analysis of our algorithm uses the adaptive
    learning graph formulation of quantum walks described in the last section.

We let $\mathrm{HG}(n, m, r)$ denote the hypergeometric distribution,
    where $n$ is the 
    total number of instances, $m$ is the number of good instances, and $r$ is the 
    number of draws without replacement.
The tail bound of this distribution is given below.

\begin{lemma}[\cite{LGNT16}] \label{lem:hg_dist} 
    Suppose $X \sim HG(n, m, r)$ with mean value $\mu = \frac{rm}{n}$, we have \begin{enumerate}
        \item for any $0 < \delta \le 1$, $\Pr\roundb{X \ge (1 + \delta)\mu} \le \exp\roundb{\frac{\mu\delta^2}{3}}$,
        \item for any $\delta > 2e - 1$, $\Pr\roundb{X > (1 + \delta)\mu} < 2^{-(1+\delta)\mu}$.
    \end{enumerate}
\end{lemma}

In the remainder of this section,
we extend the algorithm presented in \cite{LGNT16}
to simplex finding in rank-$4$ hypergraphs, proving the following theorem.

\begin{theorem} \label{thm:4_simplex_algo}
There is an adaptive learning graph algorithm for computing the $4$-simplex finding problem with $O(n^{2.4548})$ quantum queries. 
\end{theorem}

\subsection{Constructing the algorithm}

    The algorithm is based on a nested quantum walk where we load all vertices of a $4$-simplex first, then 
        load the pairs of these vertices, the triples of these vertices, and finally the $4$-hyperedges of this $4$-simplex.
    This nested quantum walk utilizes $30$ real parameters $0 \le a_i, b_{ij}, c_{ijk}, d_{ijk\ell} < 1$ 
        for $ijk\ell\in \binom{[5]}{4}$.
    For convenience of notation, we also define 15 dependent values \begin{align*}
        m_{ijk} &= b_{ij} + b_{ik} + b_{jk} - a_i - a_j - a_k,\\
        m_{ijk\ell} &= c_{ijk} + c_{ij\ell} + c_{ik\ell} + c_{jk\ell} - b_{ij} - b_{ik} - b_{i\ell} - b_{jk} - b_{j\ell} - b_{k\ell} 
            + a_i + a_j + a_k + a_\ell.
    \end{align*}
    We say that the set of parameters $\{a_i, b_{ij}, c_{ijk}, d_{ijk\ell} : ijk\ell\in \binom{[5]}{4}\}$ is 
        \emph{admissible} if the following set of (possibly strict) linear conditions hold.
    \begingroup
    \allowdisplaybreaks
    \begin{align*}
        b_{ij} \le a_i + a_j &\qquad\text{ for all } ij \in \binom{[5]}{2},\\
        c_{ijk} \le m_{ijk} &\qquad\text{ for all } ijk \in \binom{[5]}{3},\\
        d_{ijk\ell} \le m_{ijk\ell} &\qquad\text{ for all } ijk\ell \in \binom{[5]}{4},\\
        a_i - b_{ij} < 0 &\qquad\text{ for all } ij \in \binom{[5]}{2},\\
        b_{ij} - m_{ijk} < 0 &\qquad\text{ for all } ijk \in \binom{[5]}{3},\\
        c_{ijk} - m_{ijk\ell} < 0 &\qquad\text{ for all } ijk\ell \in \binom{[5]}{4},\\
        -c_{ij\ell} - c_{ik\ell} + b_{i\ell} + b_{j\ell} + b_{k\ell} - a_\ell < 0 &\qquad\text{ for all } ijk\ell \in \binom{[5]}{4}.
    \end{align*}
    \endgroup

    Suppose $G$ is a 4-uniform hypergraph defined on vertex set $V$ containing a $4$-simplex as a sub-hypergraph.
    Let $u_1, u_2, u_3, u_4, u_5$ be the vertices of this 4-simplex.
    We will use \lem{learning_graph_nested_johnson} to define an adaptive learning graph $\G$ that finds
        this $4$-simplex.
    There are $r' = 30$ levels to this walk.

    In level $i \in [5]$, we search for vertex $u_i$ using a walk over the Johnson Graph $J(n, n^{a_i})$.
    Let's denote the state of this walk by $A_i \in \P([n], n^{a_i})$.
    Let $V_i = \{v_s : s \in A_i\}$. We say $A_i$ is marked if and only if $u_i \in V_i$.
    In the context of \lem{learning_graph_nested_johnson}, the associated parameters of this level are 
        $n_i = n, k_i = n^{a_i}, \ell_i = 1$ and $I_{y, i} = \{s : v_s = u_i\}$.
    The available set $C(A_1, \dots, A_{i - 1})$ is trivially defined for this level.

    We label the next $10$ levels by pairs of indices $ij \in \binom{[5]}{2}$.
    In level $ij$ where $i < j$, we invoke a quantum walk over the Johnson Graph $J(n^{a_i + a_j}, n^{b_{ij}})$.
    Let $B_{ij} \subseteq \squareb{n^{a_i + a_j}}$ be the state of this walk and let 
        $V_{ij} = \{v_{s_i}v_{s_j} : (s_i, s_j) \in \roundb{A_i\times A_j}\squareb{B_{ij}}\}$ be the associated 
        pairs of vertices.
    We say $B_{ij}$ is marked if it satisfies the following conditions.
    \begin{enumerate}
        \item $u_iu_j\in V_{ij}$,
        \item for all $v_i\in V_i$, we have 
            $n^{b_{ij} - a_i}/2 \le |\curlyb{v_j\in V_j: v_iv_j \in V_{ij}}| \le 2n^{b_{ij} - a_i}$,
        \item for all $v_j\in V_j$, we have 
            $n^{b_{ij} - a_j}/2 \le |\curlyb{v_i\in V_i: v_iv_j \in V_{ij}}| \le 2n^{b_{ij} - a_j}$,
        \item whenever $k$ is an index such that the level $ik$ comes before $ij$ in the nested structure, we have
            $|\curlyb{v_i\in V_i: v_iv_j\in V_{ij}, v_iv_k\in V_{ik}}| \le 11n^{m_{ijk} - b_{jk}}$
            for every $v_j\in V_j, v_k\in V_k$.
    \end{enumerate}
    Note that this is formalized in the learning graph model by taking parameters 
        $n_{ij} = n^{a_i + a_j}, k_{ij} = n^{b_{ij}}, \ell_{ij} = 1$, setting $I_{y, ij} = 
            \{s : (A_i \times A_j)[s] = (s_i, s_j), v_{s_i}v_{s_j} = u_iu_j\}$.
    Define $C(A_1, \dots, B_{ij - 1}) = \curlyb{B_{ij} : \text{ Condition 2, 3, 4 holds for } B_{ij} }$.
    Here, $B_{ij - 1}$ is just denotes the state prior to $B_{ij}$.
    The following lemma is presented in \cite{LGNT16} and shows that the fraction of $B_{ij}$ for which
        condition 2, 3, or 4 doesn't hold is small.
    \begin{lemma} \label{lem:2_edge_small}
        Given marked states $A_1, \dots, B_{ij-1}$, we have
        \begin{multline*}
            \Pr[B_{ij} \not\in C(A_1, \dots, B_{ij-1})]\\
            \le O\roundb{n^{a_i}\exp\roundb{n^{a_i - b_{ij}}} + n^{a_j}\exp\roundb{n^{a_j - b_{ij}}} 
            + n^{a_i + a_j}\exp\roundb{n^{b_{jk} - m_{ijk}}}}.
        \end{multline*}
    \end{lemma}
    We will not restate its proof, but its idea is captured by the proof of \lem{3_edge_small}.
    Provided the set of parameters is admissible, we can take $\alpha(n)$ an exponentially decreasing function.

    The next $10$ levels are labeled by triples of indices $ijk \in \binom{[5]}{3}$.
    In level $ijk$ where $i < j < k$, we quantum walk over the Johnson Graph 
        $J\roundb{11 n^{m_{ijk}}, n^{c_{ijk}}}$.
    Let $C_{ijk} \subseteq \squareb{11 n^{m_{ijk}}}$ be the state of this walk and define
    \begin{align*}
        \Gamma_{ijk} &= \curlyb{(s_i, s_j, s_k) \in A_i\times A_j \times A_k : v_{s_i}v_{s_j} \in V_{ij}, 
            v_{s_i}v_{s_k} \in V_{ik}, v_{s_j}v_{s_k} \in V_{jk}}\\
        V_{ijk} &= \curlyb{v_{s_i}v_{s_j}v_{s_k} : (s_i, s_j, s_k) = \Gamma_{ijk}[s] \text{ for } 
            s \in C_{ijk}, s \le |\Gamma_{ijk}| }.
    \end{align*}
    We say $C_{ijk}$ is marked if it satisfies the following conditions.
    \begin{enumerate}
        \item $u_iu_ju_k \in V_{ij}$,
        
        \item for all $v_jv_k \in V_{jk}$, we have 
            $|\curlyb{v_i\in V_i: v_iv_jv_k \in V_{ijk}}| \le \frac{1}{6} n^{c_{ijk} - b_{jk}}$,

        \item for all $v_iv_k \in V_{ik}$, we have 
            $|\curlyb{v_j\in V_j: v_iv_jv_k \in V_{ijk}}| \le \frac{1}{6} n^{c_{ijk} - b_{ik}}$,

        \item for all $v_iv_j \in V_{ij}$, we have 
            $|\curlyb{v_k\in V_k: v_iv_jv_k \in V_{ijk}}| \le \frac{1}{6} n^{c_{ijk} - b_{ij}}$,

        \item whenever $\ell$ is an index such that the levels $ij\ell, ik\ell$ come before $ijk$ in the nested structure, 
            we have $|\curlyb{v_i\in V_i: v_iv_jv_k\in V_{ijk}, v_iv_jv_\ell\in V_{ij\ell}, v_iv_kv_\ell \in V_{ik\ell}}| 
            \le \frac{1}{11} n^{m_{ijk\ell} - c_{jk\ell}}$ for every $v_j\in V_j, v_k\in V_k, v_\ell \in V_\ell$.
    \end{enumerate}
    In the learning graph model, the associated parameters are 
        $n_{ijk} = 11 n^{m_{ijk}}, k_{ijk} = n^{c_{ijk}},\\ \ell_{ijk} = 1$. 
    We set $I_{y, ijk} = \{s: \Gamma_{ijk}[s] = (s_i, s_j, s_k), v_{s_i}v_{s_j}v_{s_k} = u_iu_ju_k\}$ and define
    \[C(A_1, \dots, C_{ijk - 1}) = \curlyb{C_{ijk} : \text{ Condition 2 to 5 holds for } C_{ijk} }\]
    assuming $A_1, \dots, C_{ijk - 1}$ are marked.
    By condition 4 of the definition of marked $B_{ik}$, we have
    \begin{equation*} 
        |\Gamma_{ijk}| = \sum_{v_jv_k \in B_{jk}} |\curlyb{v \in V_i : vv_{j} \in V_{ij} \text{ and } 
            vv_{k} \in V_{ik}}| \le n^{b_{jk}}\cdot 11 n^{m_{ijk} - b_{jk}} = 11 n^{m_{ijk}}.
    \end{equation*}
    This ensures that the certificate $u_iu_ju_k$ will not overflow and such an index $s$ exists for $I_{y, ijk}$.
    Similarly to \lem{2_edge_small}, we show that the fraction of $C_{ijk}$ for which 
        conditions 2 to 5 don't hold is also small.
    \begin{lemma} \label{lem:3_edge_small}
        Given marked states $A_1, \dots, C_{ijk-1}$, the value
        $\Pr[C_{ijk} \not\in C(A_1, \dots, C_{ijk-1})]$
        is an exponentially close to $0$ with respect to $n$, provided the set of parameters are admissible.
    \end{lemma}
    To avoid interrupting the presentation of the algorithm,
    the proofs of this lemma and the lemmas 
    in the rest of this section are presented in \app{4_simplex_lemma_proof}.

    The last $5$ levels are labeled by quadruples of indices $(i,j,k,\ell) \in \binom{[5]}{4}$.
    In level $ijk\ell$ where $i < j < k < \ell$, we invoke a quantum walk over the Johnson Graph 
        $J\roundb{\Theta(n^{m_{ijk\ell}}), n^{d_{ijk\ell}}}$.
    Let $D_{ijk\ell}\subseteq \squareb{\Theta(n^{m_{ijk\ell}})}$ be the state of this walk.
    Define
    \begin{align*}
        \Gamma_{ijk\ell} &= \{(s_i, s_j, s_k, s_\ell) \in A_i\times A_j \times A_k \times A_\ell : 
            v_{s_i}v_{s_j}v_{s_k} \in V_{ijk}, v_{s_i}v_{s_j}v_{s_\ell} \in V_{ij\ell},\\
            & \qquad \qquad \qquad v_{s_i}v_{s_k}v_{s_\ell} \in V_{ik\ell}, v_{s_j}v_{s_k}v_{s_\ell} \in V_{jk\ell}
        \}\\
        V_{ijk\ell} &= \curlyb{v_{s_i}v_{s_j}v_{s_k}v_{s_\ell} : (s_i, s_j, s_k, s_\ell) = \Gamma_{ijk\ell}[s] \text{ for } 
            s \in D_{ijk\ell}, s \le |\Gamma_{ijk\ell}| }.
    \end{align*}
    We say that $D_{ijk\ell}$ is marked if and only if $u_iu_ju_ku_\ell \in V_{ijk\ell}$.
    In the learning graph, the corresponding parameters are $n_{ijk\ell} = \Theta(n^{m_{ijk\ell}}), k_{ijk\ell} = n^{d_{ijk\ell}},
        \ell_{ijk\ell} = 1$.
    We set $I_{y, ijk\ell} = \{s\}$ where $\Gamma_{ijkl}[s] = (s_i, s_j, s_k, s_\ell)$ and 
        $v_{s_i}v_{s_j}v_{s_k}v_{s_\ell} = u_iu_ju_ku_\ell$.
    Assuming $A_1, \dots, D_{ijk\ell - 1}$ are marked.
    By condition 7 of the definition of marked $C_{ij\ell}$, we have 
    \begin{align*} 
        |\Gamma_{ijk\ell}| &= \sum_{\substack{(s_j, s_k, s_\ell) = \Gamma_{jk\ell}[s]\\ s \in C_{jk\ell}}} 
            \Big|\curlyb{v \in V_i : vv_{s_j}v_{s_k} \in V_{ijk}, vv_{s_j}v_{s_\ell} \in V_{ij\ell}, 
            vv_{s_k}v_{s_\ell}\in V_{ik\ell}} \Big|\\ 
            &\le n^{c_{jk\ell}}\cdot \frac{1}{11} n^{m_{ijk\ell} - c_{jk\ell}} 
        \\&= \Theta\roundb{n^{m_{ijk\ell}}}.
    \end{align*}
    This ensures that the certificate $u_iu_ju_ku_\ell$ will not overflow and such an index $s$ exists for $I_{y, ijk\ell}$.
    $C(A_1, \dots, D_{ijk\ell - 1})$ is trivially defined for this level.

    Let $\bA = \roundb{A_1, \dots, A_5, B_{12}, \dots, B_{45}, C_{123}, \dots, C_{345}, 
        D_{1234}, \dots, D_{2345}}$ be the sequence of states.
        The associated data structure is given by
        $\displaystyle D(\bA) := \bigcup_{ijk\ell\in \binom{[5]}{4}}V_{ijk\ell}$.
    It is important to note that a state in $\P([n_i], k_i)$ may only be partially filled.
    If any of the entries in the $A_i, B_{ij}, C_{ijk}, D_{ijk\ell}$ necessary to identify 
        a quadruple in $V_{ijk\ell}$ is missing, this quadruple will not be listed in $V_{ijk\ell}$.
    The cost of setup is $\displaystyle \bS \le \sum_{ijk\ell \in \binom{[5]}{4}}n^{d_{ijk\ell}}$.
    Once we have the query information in $D(\bA)$, it is trivial to check if $u_1, \dots, u_5$ form a $4$-simplex.
    Thus, the cost of checking $\bC$ is $0$ and it remains to find and justify the update costs based on the 
        size of the tuple of vertices we are loading.

\subsection{Analyzing the algorithm}
    In update stage $i \in [5]$, we start with beginning L-vertex $\bA = (A_1, \dots, D_{2345})$ where 
        $|A_i| = n^{a_i} - 1$.
    If we are loading $s \not\in A_i$ to $A_i$, the queries needed to update the data structure are precisely 
        the number of newly identifiable quadruples in $V_{ijk\ell}$ due to loading $s$.
    The following lemma identifies the update costs.
    \begin{lemma} \label{lem:1_edge_update}
        Let $\Gamma'_{ijk}, V'_{ijk}$ be the sets $\Gamma_{ijk}, V_{ijk}$ obtained after loading a random element 
            $s$ to $A_i$. Then 
        \[\Exp_{\bA, s} |\Gamma'_{ijk\ell} - \Gamma_{ijk\ell}| = O(n^{m_{ijk\ell} - a_i})
            \text{ and } \Exp_{\bA, s} |V'_{ijk\ell} - V_{ijk\ell}| = O(n^{d_{ijk\ell} - a_i}).\]
    \end{lemma}
    
    By the above lemma, we can conclude that
    \begin{align*}
        \bU_i &= O\roundb{\Exp_{\bA, s} |D(\dots, A_i\cup \{s\}, \dots, D_{2345}) - D(\dots, A_i, \dots, D_{2345})|}\\
        &= O\roundb{ \sum_{j, k, \ell: ijk\ell \in \binom{[5]}{4}} \Exp_{\bA, s}|V'_{ijk\ell} - V_{ijk\ell}|} 
        \\&= O\roundb{ \sum_{j, k, \ell: ijk\ell \in \binom{[5]}{4}} n^{d_{ijk\ell} - a_i}}.
    \end{align*}

    In update stage $ij \in \binom{[5]}{2}$ where $i < j$, we start with beginning L-vertex $\bA = (A_1, \dots, D_{2345})$ 
        where $|B_{ij}| = n^{b_{ij}} - 1$.
    If we are loading $s$ to $B_{ij}$, we are again looking for the newly identifiable quadruples in 
        $V_{ijk\ell}$ due to loading $s$.
    \begin{lemma} \label{lem:2_edge_update}
        Let $\Gamma'_{ijk\ell}, V'_{ijk\ell}$ be the set $\Gamma_{ijk\ell}, V_{ijk\ell}$ obtained after loading index $s$ to 
            $B_{ij}$. Then 
        \[\Exp_{\bA, v} |\Gamma'_{ijk\ell} - \Gamma_{ijk\ell}| = O(n^{m_{ijk\ell} - b_{ij}})
            \text{ and } \Exp_{\bA, v} |V'_{ijk\ell} - V_{ijk\ell}| = O(n^{d_{ijk\ell} - b_{ij}}).\]
    \end{lemma}

    The above lemma shows that
    \begin{align*}
        \bU_{ij} &= O\roundb{\Exp_{\bA, v} |D(\dots, B_{ij}\cup \{t_it_j\}, \dots) - D(\dots, B_{ij}, \dots)|}\\
        &= O\roundb{ \sum_{k, \ell: ijk\ell \in \binom{[5]}{4}} \Exp_{\bA, t_it_j}|V'_{ijk\ell} - V_{ijk\ell}|} 
        \\&= O\roundb{ \sum_{k, \ell: ijk\ell \in \binom{[5]}{4}} n^{d_{ijk\ell} - b_{ij}}}.
    \end{align*}

    In update stage $ijk \in \binom{[5]}{3}$ where $i < j < k$, we start with begining L-vertex 
        $\bA = (A_1, \dots, D_{2345})$ where $|C_{ijk}| = n^{c_{ijk}} - 1$.
    If we are loading the triple $v_iv_jv_k$ to $C_{ijk}$, we look for newly identifiable quadruples in 
        $V_{ijk\ell}$ due to loading $v_iv_jv_k$.
    \begin{lemma} \label{lem:3_edge_update}
        Let $\Gamma'_{ijk\ell}, V'_{ijk\ell}$ be the set
        $\Gamma_{ijk\ell}, V_{ijk\ell}$ obtained after loading a random triple 
            $v_iv_jv_k$ to $C_{ijk}$. Then 
        \[\Exp_{\bA, v} |\Gamma'_{ijk\ell} - \Gamma_{ijk\ell}| = O(n^{m_{ijk\ell} - c_{ijk}})
            \text{ and } \Exp_{\bA, v} |V'_{ijk\ell} - V_{ijk\ell}| = O(n^{d_{ijk\ell} - c_{ijk}}).\]
    \end{lemma}
    
    The above lemma shows that
    \begin{align*}
        \bU_{ijk} &= O\roundb{\Exp_{\bA, v} |D(\dots, C_{ijk}\cup \{s_is_js_k\}, \dots) - D(\dots, C_{ijk}, \dots)|}\\
        &= O\roundb{ \sum_{\ell: ijk\ell \in \binom{[5]}{4}} \Exp_{\bA, s_is_js_k}|V'_{ijk\ell} - V_{ijk\ell}|} 
        \\&= O\roundb{ \sum_{\ell: ijk\ell \in \binom{[5]}{4}} n^{d_{ijk\ell} - c_{ijk}}}.
    \end{align*}

    Finally, in update stage $ijk\ell \in \binom{[5]}{4}$ where $i < j < k < \ell$, the cost of update is at most $\bU_{ijk\ell} = 1$.
    By \lem{learning_graph_nested_johnson}, the query complexity of this learning graph is \begin{equation*}
        O\roundb{\bS + \sum_{i = 1}^{30}\roundb{\prod_{j=1}^i \sqrt{\frac{n_j}{k_j}}}\cdot \sqrt{k_i} \cdot \bU_i}.
    \end{equation*}
    We summarize the parameters that appear in the above complexity in \tab{4_simplex_complexity}.
    \begin{table}\begin{center}\begin{tabular}{| c | c | c | c | c |}
        \hline
        level $s$ & $i$ & $ij$ & $ijk$ & $ijk\ell$\\ \hline
        $n_s$ & $n$ & $n^{a_i + a_j}$ & $\Theta\roundb{n^{m_{ijk}}}$ & $\Theta\roundb{n^{m_{ijk\ell}}}$\\ \hline 
        $k_s$ & $n^{a_i}$ & $n^{b_{ij}}$ & $n^{c_{ijk}}$ & $n^{d_{ijk\ell}}$\\ \hline 
        $U_s$ & $\displaystyle O\roundb{\max_{j, k, \ell}n^{ d_{ijk\ell} - a_i}}$ & 
            $\displaystyle O\roundb{\max_{k, \ell}n^{ d_{ijk\ell} - b_{ij}}}$ & 
            $\displaystyle O\roundb{\max_{\ell}n^{ d_{ijk\ell} - c_{ijk}}}$ & $1$\\ \hline
    \end{tabular}\end{center}
    \caption[
        Parameters and update complexities of nested quantum walk learning graph for $4$-simplex finding.
    ]{
        Parameters and quantum query complexities of update for each level of the nested quantum walk learning graph 
        for $4$-simplex finding.
    }
    \label{tab:4_simplex_complexity}
    \end{table}

    Optimizing a linear program involving parameters $a_i, b_{ij}, c_{ijk}, d_{ijk\ell}$, the optimal complexity 
        comes down to $O(n^{2.455})$ by taking (approximate) parameter values
        \begin{align*}
            &a_1 = 0.30435, \quad a_2 = 0.65217, \quad a_3 = 0.82609, \quad a_4 = 0.91304, 
                \quad a_5 = 0.95652,\\
            &b_{12} = 0.95652, \quad b_{13} = 1.13043, \quad b_{14} = 1.21739, \quad b_{15} = 1.16579, 
                \quad b_{23} = 1.45059,\\
            &b_{24} = 1.45059, \quad b_{25} = 1.54567, \quad b_{34} = 1.49802, \quad b_{35} = 1.64032,  
                \quad b_{45} = 1.75494,\\
            &c_{123} = 1.75494, \quad c_{124} = 1.75494, \quad c_{125} = 1.75494, \quad c_{134} = 1.80237,
                \quad c_{135} = 1.84958,\\
            &c_{145} = 1.87440, \quad c_{234} = 1.95477, \quad c_{235} = 2.04985, \quad c_{245} = 2.13966,
                \quad c_{345} = 2.07817,\\
            &d_{1234} = 2.25911, \enspace  d_{1235} = 2.25911, \enspace  d_{1245} = 2.25911, \enspace  
                d_{1345} = 2.16864, \enspace  d_{2345} = 2.13966.\qedhere
        \end{align*}
This concludes the proof of \thm{4_simplex_algo}.

\if\anon1
\else
\section*{Acknowledgements}
We thank Richard Cleve and Ashwin Nayak for helpful comments.

This research is supported in part by the Natural Sciences and Engineering Research Council of Canada (NSERC), DGECR-2019-00027 and RGPIN-2019-04804, as well as a CGS-M award.\footnote{Cette recherche a été financée par le Conseil de recherches en sciences naturelles et en génie du Canada (CRSNG),
DGECR-2019-00027 et RGPIN-2019-04804.}
\fi

\phantomsection\addcontentsline{toc}{section}{References} 
\renewcommand{\UrlFont}{\ttfamily\small}
\let\oldpath\path
\renewcommand{\path}[1]{\small\oldpath{#1}}
\emergencystretch=1em 
\printbibliography

\appendix
\section{Proofs of lemmas for 4-simplex finding}
\label{app:4_simplex_lemma_proof}

In this appendix, we will prove the lemmas that appears in the proof of
    \thm{4_simplex_algo}.

\begin{proof}[Proof of \lem{3_edge_small}]
    Fix $v_jv_k \in V_{jk}$, and define the set
    $S_1 = \curlyb{v_i \in V_i : v_iv_jv_k \in V_{ijk}}$.
    For any $v_i \in V_i$, we have \begin{align} \label{eqn:s1_prob}
        \Pr[v_iv_jv_k \in V_{ijk}] &= \Pr[v_iv_jv_k \in V_{ijk} | v_iv_j \in V_{ij}, v_iv_k \in V_{ik}]
            \Pr[v_iv_j \in V_{ij}]\Pr[v_iv_k \in V_{ik}] \nonumber\\
        &= \frac{1}{11} n^{c_{ijk} - m_{ijk}}\cdot n^{b_{ij} - a_i - a_j}\cdot n^{b_{ik} - a_i - a_k}
        = \frac{1}{11} n^{c_{ijk} - b_{jk} - a_i}.
    \end{align}
    Therefore, $|S_1|$ is a random variable with hypergeometric distribution 
        $\mathrm{HG}(n^{a_i+b_{jk}}, n^{a_i}, n^{c_{ijk}}/11)$.
    It has mean value $\frac{1}{11} n^{c_{ijk} - b_{jk}}$ and by \lem{hg_dist} (1), \begin{equation}
        \Pr\squareb{|S_1| \ge \frac{1}{6} n^{c_{ijk} - b_{jk}}} \le \exp\roundb{-\frac{1}{55}n^{c_{ijk} - b_{jk}}}.
    \end{equation}
    We can prove a similar bound for conditions 3 and 4.
    By the union bound, we see that \begin{multline} \label{eqn:cond_2_3_4}
        \Pr[\text{Conditions } 2, 3, 4 \text{ don't hold}] \le \\
            n^{b_{jk}}\exp\roundb{-\frac{1}{55}n^{c_{ijk} - b_{jk}}} + 
            n^{b_{ik}}\exp\roundb{-\frac{1}{55}n^{c_{ijk} - b_{ik}}} + 
            n^{b_{ij}}\exp\roundb{-\frac{1}{55}n^{c_{ijk} - b_{ij}}}.
    \end{multline}

    For condition 5, fix vertices $v_j \in V_j, v_k \in V_k, v_\ell\in V_\ell$ such that 
        $v_jv_k \in V_{jk}, v_jv_\ell \in V_{j\ell},\\ v_kv_\ell \in V_{k\ell}$.
    Define $S'_1 := \curlyb{v_i \in V_i : v_iv_jv_\ell \in V_{ij\ell}}$ and 
        $S_2 := \curlyb{v_i \in S'_1 : v_iv_kv_\ell \in V_{ik\ell}}$.
    Note that for any $v_i \in V_i$, we have \begin{align} \label{eqn:s2_prob}
        \Pr[v_iv_kv_\ell \in V_{ik\ell} | v_i \in S'_1] &= \Pr[v_iv_kv_\ell \in V_{ik\ell} | v_iv_\ell \in V_{i\ell}, 
            v_iv_k \in V_{ik}]\Pr[v_iv_k \in V_{ik}] \nonumber\\
        &= \frac{1}{11} n^{c_{ik\ell} - m_{ik\ell}}\cdot n^{b_{ik} - a_i - a_k}
        = \frac{1}{11} n^{c_{ik\ell} - b_{kl} - b_{i\ell} + a_\ell}.
    \end{align}
    Hence $|S_2|$ follows the distribution 
        \[\mathrm{HG}\roundb{|S'_1|\cdot |\Gamma_{jk\ell}|,\ |S'_1|,\ \frac{1}{11} |S'_1|\cdot |\Gamma_{jk\ell}|
        \cdot n^{c_{ik\ell} - b_{k\ell} - b_{i\ell} + a_\ell}}.\]
    It has mean $\frac{1}{11} |S'_1| n^{c_{ik\ell} - b_{k\ell} - b_{i\ell} + a_\ell}$. 
    Since we assume $C_{ij\ell}$ is marked, we have $|S'_1| \le \frac{1}{6}n^{c_{ij\ell}-b_{j\ell}}$.
    Applying \lem{hg_dist} (2) with $\delta = \bigslash{n^{c_{ij\ell} - b_{j\ell}}}{|S'_1|} - 1 > 2e - 1$, we have
        \begin{equation} \label{eqn:s2_bound}
        \Pr\squareb{|S_2| > \frac{1}{11}n^{c_{ij\ell} + c_{ik\ell} - b_{i\ell} - b_{j\ell} - b_{k\ell} + a_\ell}}
            \le \exp\roundb{-\frac{\log 2}{11}n^{c_{ij\ell} + c_{ik\ell} - b_{i\ell} - b_{j\ell} - b_{k\ell} + a_\ell}}.
        \end{equation}
    Finally, define \begin{equation*}
        S_3 := \curlyb{v_i\in V_i: v_iv_jv_k\in V_{ijk}, v_iv_jv_\ell\in V_{ij\ell}, v_iv_kv_\ell \in V_{ik\ell}} 
        = \curlyb{v_i \in S_2 : v_iv_jv_k \in V_{ijk}}.
    \end{equation*}
    For any $v_i \in V_i$, we have \begin{align} \label{eqn:s3_prob}
        \Pr\squareb{v_iv_jv_k \in V_{ijk} | v_i \in S_2}  = 
        \Pr\squareb{v_iv_jv_k \in V_{ijk} | v_iv_j \in V_{ij}, v_iv_k \in V_{ik}} = 
        \frac{1}{11}n^{c_{ijk} - m_{ijk}}.
    \end{align}
    Thus $|S_3|$ follows the distribution 
        \[\mathrm{HG}\roundb{|S_2| \cdot |\Gamma_{jk\ell}|,\ |S_2|,\ \frac{1}{11} |S_2|\cdot |\Gamma_{jk\ell}|
        \cdot n^{c_{ijk} - m_{ijk}}},\]
    which has mean $\frac{1}{11} |S_2| n^{c_{ijk} - m_{ijk}}$. 
    Under the condition that $|S_2| \le \frac{1}{11}n^{c_{ij\ell} + c_{ik\ell} - b_{i\ell} - b_{j\ell} - b_{k\ell} + a_\ell}$,
        we apply \lem{hg_dist} (2) with $\delta = 
        \bigslash{n^{c_{ik\ell} + c_{ij\ell} - b_{i\ell} - b_{j\ell} - b_{k\ell} + a_\ell}}{|S_2|} - 1 > 2e - 1$, getting
    \begin{equation} \label{eqn:s3_bound}
    \Pr\squareb{|S_3| > \frac{1}{11}n^{m_{ijk\ell} - c_{jk\ell}} \middle| \enspace
        |S_2| \le \frac{1}{11}n^{c_{ij\ell} + c_{ik\ell} - b_{i\ell} - b_{j\ell} - b_{k\ell} + a_\ell}}
        \le \exp\roundb{-\frac{\log 2}{11}n^{m_{ijk\ell} - c_{jk\ell}}}.
    \end{equation}
    Combining equations (\ref{eqn:s2_bound}) and (\ref{eqn:s3_bound}), we get 
    \begin{equation}
    \label{eqn:cond_5}
        \Pr[\text{Condition } 5 \text{ fails}] \le
            \Theta(n^{m_{jk\ell}}) \squareb{
                \exp\roundb{n^{c_{ij\ell} + c_{ik\ell} - b_{i\ell} - b_{j\ell} - b_{k\ell} + a_\ell}} + 
                \exp\roundb{n^{m_{ijk\ell} - c_{jk\ell}}}
            }.
    \end{equation}
    The statement of this lemma follows from equations (\ref{eqn:cond_2_3_4}), (\ref{eqn:cond_5}), 
        and the union bound.
\end{proof}

\begin{proof}[Proof of \lem{1_edge_update}]
    By the definition of the set $\Gamma_{ijk\ell}$, we see that \begin{align*}
        \Exp_{\bA, s} &\squareb{|\Gamma'_{ijk\ell} - \Gamma_{ijk\ell}|} = \Theta\roundb{\frac{1}{n}}\sum_{s}\Exp_{\bA}
        \squareb{|\Gamma'_{ijk\ell} - \Gamma_{ijk\ell}|}\\
        &\le \Theta\roundb{\frac{1}{n}}\sum_{s}\Exp_{\bA}
            |\{v_jv_kv_\ell \in V_{jk\ell}: v_sv_jv_k \in V_{ijk}, v_sv_jv_\ell \in V_{ij\ell}, v_sv_kv_\ell \in V_{ik\ell}\}|\\
        &= \Theta\roundb{\frac{1}{n}}\sum_{s}\sum_{v_jv_kv_\ell \in V_{jk\ell}}\Pr(v_sv_jv_k \in V_{ijk}, 
            v_sv_jv_\ell \in V_{ij\ell}, v_sv_kv_\ell \in V_{ik\ell})\\
        &= \Theta\roundb{n^{c_{jk\ell}}} \Theta\roundb{n^{c_{ij\ell} - b_{j\ell} - a_i}} 
            \Theta\roundb{n^{c_{ik\ell} - b_{k\ell} - b_{i\ell} + a_\ell}} \Theta\roundb{n^{c_{ijk} - m_{ijk}}} 
        \\&= \Theta\roundb{n^{m_{ijk\ell} - a_i}}.
    \end{align*}
    The third equality is a consequence of equations (\ref{eqn:s1_prob}), (\ref{eqn:s2_prob}), (\ref{eqn:s3_prob}).
    Finally, since every tuple in $\Gamma_{ijk\ell}$ becomes a quadruple in $V_{ijk\ell}$ only with probability 
        $\Theta\roundb{n^{d_{ijk\ell} - m_{ijk\ell}}}$, we have 
    \[\Exp_{\bA, s} |V'_{ijk\ell} - V_{ijk\ell}| = 
        O\roundb{n^{d_{ijk\ell} - m_{ijk\ell}}
            \cdot \Exp_{\bA, s} |\Gamma'_{ijk\ell} - 
        \Gamma_{ijk\ell}|} = O(n^{d_{ijk\ell} - a_i}).\qedhere\]
\end{proof}

\vspace*{1em}
\begin{proof}[Proof of \lem{2_edge_update}]
    Fixing $v_i \in V_i, v_j\in V_j$, $v_kv_\ell \in V_{k\ell}$, we see that \begin{align} \label{eqn:2_edge_precom}
        &\Pr\roundb{v_iv_kv_\ell \in V_{ik\ell}, v_jv_kv_\ell \in V_{jk\ell}} \nonumber\\
        =\ &\Pr\roundb{v_iv_kv_\ell \in V_{ik\ell}, v_jv_kv_\ell \in V_{jk\ell} | v_iv_k \in V_{ik}, v_jv_k \in V_{jk}}
            \Pr\roundb{v_iv_k \in V_{ik}, v_jv_k \in V_{jk}} \nonumber\\
        =\ &\Theta\roundb{n^{c_{ik\ell} - m_{ik\ell}}} \cdot n^{b_{ik} - a_i - a_k} \cdot n^{b_{i\ell} - a_i - a_\ell}
            \cdot \Theta\roundb{n^{c_{jk\ell} - m_{jk\ell}}} \cdot n^{b_{jk} - a_j - a_k} \cdot n^{b_{j\ell} - a_j - a_\ell} 
            \nonumber\\
        =\ &\Theta\roundb{n^{c_{ik\ell} + c_{jk\ell} - 2b_{k\ell} - a_i - a_j}}.
    \end{align}
    Similar to the proof of \lem{1_edge_update}, we write
    \begin{align*}
        &\Exp_{\bA, v} |\Gamma'_{ijk\ell} - \Gamma_{ijk\ell}| = \Theta\roundb{\frac{1}{n^{a_i+a_j}}}\sum_{t_it_j}
            \Exp_{\bA} |\Gamma'_{ijk\ell} - \Gamma_{ijk\ell}|\\
        \le\  &\Theta\roundb{\frac{1}{n^{a_i+a_j}}}\sum_{t_it_j}\Exp_{\bA}
            |\{v_kv_\ell \in V_{k\ell}: v_iv_kv_\ell \in V_{ik\ell}, v_jv_kv_\ell \in V_{jk\ell},
                v_iv_jv_k \in V_{ijk}, v_iv_jv_\ell \in V_{ij\ell}\}|\\
        =\  &\Theta\roundb{\frac{1}{n^{a_i+a_j}}}\sum_{t_it_j}\sum_{v_kv_\ell \in V_{k\ell}}
            \Pr(v_iv_jv_k, v_iv_jv_\ell| v_iv_kv_\ell, v_jv_kv_\ell)\Pr(v_iv_kv_\ell, v_jv_kv_\ell)\\
        =\  &\Theta(n^{b_{k\ell}}) \Theta(n^{c_{ijk} - m_{ijk} + c_{ij\ell} - m_{ij\ell}})
            \Theta\roundb{n^{c_{ik\ell} + c_{jk\ell} - 2b_{k\ell} - a_i - a_j}}
        = \Theta\roundb{n^{m_{ijk\ell} - b_{ij}}}.
    \end{align*}
    The second equality is a consequence of equation (\ref{eqn:2_edge_precom}).
    Since every tuple in $\Gamma_{ijk\ell}$ becomes a quadruple in $V_{ijk\ell}$ only 
    with probability 
        $\Theta\roundb{n^{d_{ijk\ell} - m_{ijk\ell}}}$, we have
    \[\Exp_{\bA, v} |V'_{ijk\ell} - 
        V_{ijk\ell}| = O\roundb{n^{d_{ijk\ell} - m_{ijk\ell}} \cdot \Exp_{\bA, v} 
        |\Gamma'_{ijk\ell} - \Gamma_{ijk\ell}|} = O(n^{d_{ijk\ell} - b_{ij}}).\qedhere\]
\end{proof}

\begin{proof}[Proof of \lem{3_edge_update}]
    Similar to the proof of \lem{1_edge_update}, we write
    \begin{align*}
        &\Exp_{\bA, v} |\Gamma'_{ijk\ell} - \Gamma_{ijk\ell}|
        = \frac{1}{|\Gamma_{ijk}|}\sum_{s_is_js_k \in \Gamma_{ijk}}
            \Exp_{\bA} |\Gamma'_{ijk\ell} - \Gamma_{ijk\ell}|\\
        \le\  &\frac{1}{|\Gamma_{ijk}|}\sum_{s_is_js_k \in \Gamma_{ijk}}\Exp_{\bA}
            |\{v_\ell \in V_\ell: v_{s_i}v_{s_j}v_\ell \in V_{ij\ell}, v_{s_i}v_{s_k}v_\ell \in V_{ik\ell}, 
            v_{s_j}v_{s_k}v_\ell \in V_{jk\ell}\}|\\
        =\  &\frac{1}{|\Gamma_{ijk}|}\sum_{s_is_js_k \in \Gamma_{ijk}}\sum_{v_\ell \in V_\ell}
            \Pr[v_{s_i}v_{s_j}v_\ell \in V_{ij\ell}, v_{s_i}v_{s_k}v_\ell \in V_{ik\ell}, 
            v_{s_j}v_{s_k}v_\ell \in V_{jk\ell}]\\
        =\  &\Theta(n^{a_\ell}) \Theta\roundb{n^{c_{ij\ell} - b_{ij} - a_\ell}} 
            \Theta\roundb{n^{c_{ik\ell} - b_{ik} - b_{i\ell} + a_i}} \Theta\roundb{n^{c_{jk\ell} - m_{jk\ell}}} \\
        =\ &\Theta\roundb{n^{m_{ijk\ell} - c_{ijk}}}.
    \end{align*}
    The third equality is again a consequence of equations (\ref{eqn:s1_prob}), (\ref{eqn:s2_prob}), and (\ref{eqn:s3_prob}).
    Since every tuple in $\Gamma_{ijk\ell}$ becomes a quadruple in $V_{ijk\ell}$ only with probability 
        $\Theta\roundb{n^{d_{ijk\ell} - m_{ijk\ell}}}$, we have 
    \[\Exp_{\bA, v} |V'_{ijk\ell} - V_{ijk\ell}| = O\roundb{n^{d_{ijk\ell} - m_{ijk\ell}} 
        \cdot \Exp_{\bA, v} |\Gamma'_{ijk\ell} - \Gamma_{ijk\ell}|}  
        = O(n^{d_{ijk\ell} - c_{ijk}}).\qedhere\]
\end{proof}

\end{document}